\documentclass[a4paper,10pt]{article}
\usepackage[utf8]{inputenc}

\title{Pricing and hedging game options in currency models with proportional transaction costs}
\author{Alet Roux\thanks{Department of Mathematics, University of York, Heslington, YO10 5DD, United Kingdom. Email: alet.roux@york.ac.uk}}

\usepackage[longnamesfirst]{natbib}
\usepackage{enumerate}

\usepackage{amsfonts}
\usepackage{amsthm}
\usepackage{amsmath}

\usepackage{needspace}
\usepackage{tabularx}

\DeclareMathOperator{\conv}{conv}
\DeclareMathOperator{\successors}{succ}
\DeclareMathOperator{\cone}{cone}

\theoremstyle{plain}
\newtheorem{theorem}{Theorem}
\newtheorem{proposition}[theorem]{Proposition}
\newtheorem{lemma}[theorem]{Lemma}

\theoremstyle{definition}
\newtheorem{definition}[theorem]{Definition}
\newtheorem{example}[theorem]{Example}
\newtheorem{construction}[theorem]{Construction}

\theoremstyle{remark}
\newtheorem{remark}[theorem]{Remark}

\numberwithin{theorem}{section}
\numberwithin{equation}{section}

\usepackage{xcolor}

\usepackage{makeidx}
\makeindex

\usepackage{hyperref}

\usepackage{tikz}
\usetikzlibrary{patterns}

\tikzstyle{axis}=[]

\newlength{\leveldistance}
\tikzstyle{fit level distance to}=[execute at begin picture={\settowidth{\leveldistance}{#1}\addtolength{\leveldistance}{1.5in}},execute at end picture={\setlength{\leveldistance}{1.5in}}]
\tikzstyle{override level distance}=[execute at begin picture={\setlength{\leveldistance}{#1}},execute at end picture={\setlength{\leveldistance}{1.5in}}]

\newlength{\defaultsiblingdistance} 
\newlength{\siblingdistance} 
\tikzstyle{max label lines}=[execute at begin picture={\setlength{\siblingdistance}{#1\baselineskip}\addtolength{\siblingdistance}{\baselineskip}}, execute at end picture={\setlength{\siblingdistance}{\defaultsiblingdistance}}]

\newcommand{\defaultsiblings}{2}
\newcommand{\siblings}{\defaultsiblings}
\tikzstyle{siblings}=[execute at begin picture={\renewcommand{\siblings}{#1}}, execute at end picture={\renewcommand{\siblings}{\defaultsiblings}}]
\tikzstyle{binary}=[siblings=2]

\tikzstyle{model}=[every node/.style={draw},grow'=right,->,
level/.style={level distance=\leveldistance, sibling distance=\siblingdistance}]
\tikzstyle{twostep}=[level 1/.style={level distance=\leveldistance, sibling distance=\siblings\siblingdistance}]

\newlength{\figurewidth}
\newlength{\figureheight}
\usepackage{dcolumn}
\newcolumntype{d}[1]{D{.}{.}{#1}}

\newcommand{\twovector}[2]{\left(#1,#2\right)}
\newcommand{\threevector}[3]{\left(#1,#2,#3\right)}

\begin{document}

\maketitle

\begin{abstract}
The pricing, hedging, optimal exercise and optimal cancellation of game or Israeli options are considered in a multi-currency model with proportional transaction costs. Efficient constructions for optimal hedging, cancellation and exercise strategies are presented, together with numerical examples, as well as probabilistic dual representations for the bid and ask price of a game option.


\end{abstract}

\section{Introduction}

The study of game options (also called Israeli options) date back to the seminal work of \citet{kifer2000}; the recent survey paper by \citet{Kifer2013} provides a complete chronology and literature review.
In addition to being of interest as derivative securities in their own right, game options have also played an important role in the study of other derivatives, for example callable options \cite[e.g.][]{Kuhn_Kyprianou2007} and convertible bonds \cite[e.g.][]{Kallsen_Kuhn2005,Bielecki_Crepey_Jeanblanc_Rutkowski2008,Wang_Jin2009}.

A game option is a contract between a writer (the \emph{seller}) and the holder (the \emph{buyer}) whereby a pre-specified payoff is delivered by the seller to the buyer at the earliest of the exercise time (chosen by the buyer) and the cancellation time (chosen by the seller). If the game option is cancelled before or at the same time as being exercised, then the seller also pays a cancellation penalty to the buyer. A game option is thus essentially an American option with the additional provision that the seller can cancel the option at any time before expiry, thus forcing early exercise at a price (the penalty). In practice, this feature tends to reduce 
costs for both the seller and the buyer, which makes game options an attractive alternative to American options.

It has been well observed that arbitrage pricing of European and American options in incomplete friction-free models and models with proportional transaction costs result in a range of arbitrage-free prices, bounded from below by the \emph{bid price} and from above by the \emph{ask price} \cite[see e.g.][]{follmer_schied2002,bensaid_lesne_pages_scheinkman1992,chalasani_jha2001,Roux_Zastawniak2015}. The same holds true for game options \citep{Kallsen_Kuhn2005,Kifer2013a}.

The pricing and hedging of game options in the presence of proportional transaction costs also share a number of other important properties with their European and American counterparts. (The properties for European and American options mentioned below were all established by \citet{Roux_Zastawniak2015} in a similar technical setting to the present paper.) Firstly, similar to European options, the hedging of game options is symmetric in the sense that the hedging problem for the buyer is exactly the same as the hedging problem for the seller (of a different game option with related payoff). \citet{Kifer2013a} observed this property in a two-asset model.

\citet{Kifer2013a} also showed that the probabilistic dual representations of the bid and ask prices of game options contain so-called randomised stopping times, a feature shared with the ask price of an American option \cite[for which it was first observed by][]{chalasani_jha2001}. Randomised (or mixed) stopping times have been studied by \citet{baxter_chacon1977} and many others, primarily as an aid to show the existence and properties of optimal ordinary stopping times. Randomised stopping times can be thought of as convex combinations of ordinary stopping times in a well-defined sense. The reason for the appearance of randomised stopping times in the probabilistic dual representations of the bid and ask prices is that, in the presence of transaction costs, the most expensive exercise (cancellation) strategy for the seller (buyer) of a game option to hedge against is not necessarily the same as the exercise (cancellation) strategy that is most attractive to the buyer (seller). As a result, it generally costs the seller (buyer) more to hedge against all exercise (cancellation) strategies than against the best exercise (cancellation) strategy for the buyer (seller). It turns out that the seller (buyer) must in effect be protected against a certain randomised exercise (cancellation) time.

Furthermore, similar to a long American option (i.e.~the buyer's case), the pricing and hedging problems for both the buyer and seller of game option are inherently non-convex. Thus ideas beyond convex duality are needed to study these problems. Nevertheless, the link between game options and short American options (i.e.~the seller's case, a convex problem) means that convex duality methods still have an important role to play in establishing the probabilistic dual representations.

In this paper we consider the pricing and hedging of game options in the num\'eraire-free discrete-time model of foreign exchange markets introduced by \citet{kabanov1999}, where proportional transaction costs are modelled as bid-ask spreads between currencies. This model has been well studied by \citet{kabanov_stricker2001b,kabanov_rasonyi_stricker2002,schachermayer2004} and others \cite[see also][]{Kabanov_Safarian2009}. 

The main aims of our work are twofold. Firstly, we present constructive algorithms for computing optimal exercise and cancellation times together with optimal hedging strategies for both the buyer and seller of a game option in this model. The algorithmic constructions in this paper are closely related to previously developed algorithms for the pricing and hedging of European and American options under proportional transaction costs \cite[see e.g.][]{Loehne_Rudloff2014,Roux_Zastawniak2009,Roux_Zastawniak2015}. These existing constructions yield efficient numerical algorithms; in particular they are known to price path-independent options in polynomial time in recombinant models (which typically have exponentially-sized state spaces). Numerical examples that illustrate the constructions are provided. Secondly, we establish probabilistic dual representations for the bid and ask prices of game options. In both these contributions we extend the recent results of \citet{Kifer2013a} for game options from two-asset to multi-asset models. Our proofs are rigorous, thus closing two gaps in the arguments of \citet{Kifer2013a}; see Remark~\ref{rem:expiration-date}, the comments below Proposition~\ref{prop:buyer-seller-symmetry} and Example~\ref{ex:counter} for further details.

The methods used in this paper come from convex analysis and dynamic programming, and in particular we will use recent results from \citet{Roux_Zastawniak2015} for an American option with random expiration date. The restriction to finite state space is motivated by the desire to produce computationally efficient algorithms for pricing and hedging. The restriction to discrete time is justified by a recent negative result by \citet{Dolinsky2013} that the super-replication price of a game option in continuous time under proportional transaction costs is the initial value of a trivial buy-and-hold superhedging strategy.

The structure of this paper is as follows. Section~\ref{sec:preliminaries} specifies the currency model with proportional transaction costs, and reviews various notions concerning randomised stopping times and approximate martingales. The main algorithms for pricing and hedging together with theoretical results for the seller's and buyer's position are presented in Section~\ref{sec:main}, with the proofs of all results deferred to Section~\ref{sec:proofs}. Section~\ref{sec:numerical} concludes the paper with three numerical examples.

\section{Preliminaries}
\label{sec:preliminaries}

\subsection{Proportional transaction costs}

 The num\'eraire free currency model of \citet{kabanov1999} has discrete trading dates $t=0,\ldots,T$\index{T@$T$} and is based on a finite probability space $(\Omega,\mathcal{F},P)$\index{Omega@$\Omega$}\index{F@$\mathcal{F}$}\index{P@$P$} with filtration $(\mathcal{F}_t)_{t=0}^T$.\index{Ft@$\mathcal{F}_t$} The model contains $d$ currencies (or \emph{assets}), and at any time $t$, one unit of currency $j=1,\ldots,d$ may be obtained by exchanging $\pi^{ij}_t>0$\index{piijt@$\pi^{ij}_t$} units of currency~$i=1,\ldots,d$. We assume that $\pi^{ii}_t=1$ for $i=1,\ldots,d$, i.e.~every currency may be freely exchanged for itself.
 
  Assume that the filtration $(\mathcal{F}_t)_{t=0}^T$ is generated by $(\pi^{ij}_t)_{t=0}^T$ for $i,j=1,\ldots,d$, and assume for simplicity that $\mathcal{F}_0=\{\emptyset,\Omega\}$, that $\mathcal{F}_T=\mathcal{F}=2^\Omega$ and that $P(\omega)>0$ for all $\omega\in\Omega$. Write $\mathcal{L}_\tau$\index{Lt@$\mathcal{L}_\tau$} for the family of $\mathcal{F}_\tau$-measurable $\mathbb{R}^d$-valued random variables for every stopping time $\tau$, and write $\mathcal{L}_\tau^+$\index{Ltplus@$\mathcal{L}_t^+$} for the family of non-negative random variables in $\mathcal{L}_\tau$.
  
  Let $\Omega_t$\index{Omegat@$\Omega_t$} be the set of atoms of $\mathcal{F}_t$ for $t=0,\ldots,T$. The elements of $\Omega_t$ are called the \emph{nodes}\index{node} of the model at time $t$. A node $\nu\in\Omega_{t+1}$ is called a \emph{successor} to a node $\mu\in\Omega_t$ if $\nu\subset\mu$. The collection of successors of $\mu$ is denoted $\successors \mu$\index{succ}. We shall implicitly and uniquely identify random variables $f$ in $\mathcal{L}_t$ with functions $\mu\mapsto f^\mu\in\mathbb{R}^d$ on~$\Omega_t$, and likewise every set $A\in\mathcal{L}_t$ that we will consider will be implicitly and uniquely defined by a set-valued mapping $\mu\mapsto A^\mu\subseteq\mathbb{R}^d$ on $\Omega_t$ such that
  \[
   A = \left\{f\in\mathcal{L}_t:f^\mu\in A^\mu \text{ for all }\mu\in\Omega_t\right\}.
  \]
 
 A portfolio $x=(x^1,\ldots,x^d)\in\mathcal{L}_\tau$ is called \emph{solvent}\index{solvency} at a stopping time $\tau$ if it can be exchanged into a portfolio in~$\mathcal{L}_\tau^+$ without any additional investment, i.e.~if there exist non-negative $\mathcal{F}_\tau$-measurable random variables~$\beta^{ij}$ for $i,j=1,\ldots,d$ such that
\[
 x^i + \sum_{j=1}^d \beta^{ji} - \sum_{j=1}^d\beta^{ij}\pi^{ij}_\tau \ge 0 \text{ for all } i=1,\ldots,d.
\]
Write $\mathcal{K}_\tau$\index{Kt@$\mathcal{K}_\tau$} for the family of solvent portfolios at time $\tau$; then the \emph{solvency cone}\index{solvency cone} $\mathcal{K}_\tau$ is the convex cone generated by the canonical basis\index{canonical basis} $e^1,\ldots,e^d$\index{ei@$e^i$} of $\mathbb{R}^d$ and the vectors $\pi^{ij}_\tau e^i-e^j$ for $i,j=1,\ldots,d$. Observe that $\mathcal{K}_\tau$ is a polyhedral cone, hence closed.

A \emph{self-financing trading strategy}\index{trading strategy} $y=(y_t)_{t=0}^T$ is an $\mathbb{R}^d$-valued predictable process with initial endowment $y_0\in\mathcal{L}_0$ satisfying $y_t - y_{t+1}\in\mathcal{K}_t$ for all $t=0,\ldots,T-1$. Denote the family of self-financing trading strategies by $\Phi$\index{Phi@$\Phi$}.

A self-financing trading strategy $y=(y_t)\in\Phi$ is called an \emph{arbitrage opportunity} if $y_0=0$ and there exists some $x\in\mathcal{L}^+_T\setminus\{0\}$ such that $y_T-x\in\mathcal{K}_T$. This definition of arbitrage is consistent with (though formally different to) that of \citet{schachermayer2004} and \citet{kabanov_stricker2001b} (who called it weak arbitrage). 

For any non-empty convex cone $A\subseteq\mathbb{R}^d$, write $A^\ast$\index{$\cdot^\ast$} for the \emph{positive polar} of~$A$,~i.e.
\[
 A^\ast := \{x\in\mathbb{R}^d:x\cdot y\ge0\text{ for all }y\in A\}.
\]
\begin{theorem}[\citet{kabanov_stricker2001b}] \label{th:ftap}
 The model is free of arbitrage if and only if there exists a probability measure~$\mathbb{P}$ equivalent to~$P$ and an $\mathbb{R}^d$-valued $\mathbb{P}$-martingale $S=(S_t)_{t=0}^T$ such that
 \[
  S_t\in\mathcal{K}^\ast_t\setminus\{0\}
 \text{ for }t=0,\ldots,T.\]
\end{theorem}

Any pair $(\mathbb{P},S)$ satisfying the conditions of Theorem~\ref{th:ftap} is called an \emph{equivalent martingale pair}. Denote the family of equivalent martingale pairs by $\mathcal{P}$; then $\mathcal{P}\neq\emptyset$ in the absence of arbitrage.

\begin{remark}
Theorem~\ref{th:ftap} and the other results in this paper can equivalently be formulated in terms of consistent pricing processes $(Z_t)_{t=0}^T$ where \[Z_t=S_t\mathbb{E}_P\left(\left.\frac{d\mathbb{P}}{dP}\right|\mathcal{F}_t\right)\text{ for all }t=0,\ldots,T.\] \citet[pp.~24--25]{schachermayer2004} provides further details on this equivalence.
\end{remark}

Assume for the remainder of this paper that the model contains no arbitrage.

\subsection{Randomised stopping times}

\begin{definition}[Randomised stopping time]
A \emph{randomised (or mixed) stopping time}\index{randomised stopping time} \mbox{$\chi=(\chi_t)_{t=0}^T$} is an adapted nonnegative process satisfying
\[
 \sum_{t=0}^T\chi_t = 1.
\]
\end{definition}
Denote the set of randomised stopping times by $\mathcal{X}$\index{Y@$\mathcal{X}$}, and the set of (ordinary) stopping times with values in $0,\ldots,T$ by $\mathcal{T}$\index{T@$\mathcal{T}$}. Every stopping time $\tau\in\mathcal{T}$ corresponds to a randomised stopping time $\chi^\tau=(\chi^\tau_t)_{t=0}^T$ defined as\index{chitau@$\chi^\tau$}\index{$\cdot^\tau$}
\[
 \chi^\tau_t:=\mathbf{1}_{\{\tau=t\}}\text{ for }t=0,\ldots,T,
\]
where $\mathbf{1}$\index{1@$\mathbf{1}$} is the indicator function on $\Omega$. The set $\mathcal{X}$ is the convex hull of \mbox{$\{\chi^\tau:\tau\in\mathcal{T}\}$} 
and so $\mathcal{X}$ can be thought of as the linear relaxation of the set of ordinary stopping times in this sense.

Fix any process $A=(A_t)_{t=0}^T$ and randomised stopping time $\chi\in\mathcal{X}$. Define the processes $\chi^\ast=(\chi^\ast_t)_{t=0}^T$\index{chiast@$\chi^\ast$} and $A^{\chi\ast}=(A^{\chi\ast}_t)_{t=0}^T$\index{$\cdot^{\chi\ast}$} by
\begin{align*}
 \chi^\ast_t&:=\sum_{s=t}^T\chi_s, & A^{\chi\ast}_t&:=\sum_{s=t}^T\chi_sA_s
\end{align*}
for $t=0,\ldots,T$.
 For convenience also define $\chi^\ast_{T+1}:=0$ and \mbox{$A^{\chi\ast}_{T+1}:=0$}. Observe that $\chi^\ast$ is a predictable process since
\[
 \chi^\ast_t = 1-\sum_{s=0}^{t-1}\chi_s\text{ for }t=1,\ldots,T.
\]
The \emph{value of $A$ at $\chi$} is defined as\index{$\cdot_\chi$}
\[
 A_\chi :=A^{\chi\ast}_0= \sum_{t=0}^T\chi_tA_t.\]

Observe that if $\chi=\chi^\tau$ for some $\tau\in\mathcal{T}$, then 
\begin{align*}
 \chi^{\tau\ast}_t &= \mathbf{1}_{\{\tau\ge t\}}, & A^{\chi^{\tau}\ast}_t&=\sum_{s=t}^T\mathbf{1}_{\{\tau=s\}}A_s = A_\tau\mathbf{1}_{\{\tau\ge t\}}
\end{align*}
for $t=0,\ldots,T$, and in particular $A_{\chi^\tau}=A_\tau$.

\begin{definition}[Approximate martingale pair] \label{def:approx-martingale}
Fix any $\chi\in\mathcal{X}$. A pair $(\mathbb{P},S)$ consisting of a probability measure~$\mathbb{P}$ and an adapted $\mathbb{R}^d$-valued process $S$ is called a \emph{$\chi$-approximate martingale pair} if 
\begin{align*}
S_t&\in\mathcal{K}^\ast_t\setminus\{0\}, & \mathbb{E}_\mathbb{P}(S^{\chi_\ast}_{t+1}|\mathcal{F}_t)&\in\mathcal{K}^\ast_t
\end{align*}
for all $t=0,\ldots,T$. If $\mathbb{P}$ is in addition equivalent to $P$, then $(\mathbb{P},S)$ is called a \emph{$\chi$-approximate equivalent martingale pair}.\index{approximate equivalent martingale pair} 
\end{definition}
Denote the family of $\chi$-approximate equivalent pairs $(\mathbb{P},S)$ by~$\mathcal{P}(\chi)$\index{Pchi@$\mathcal{P}(\chi)$} and the set of $\chi$-approximate pairs by $\bar{\mathcal{P}}(\chi)$.\index{Pchi@$\bar{\mathcal{P}}(\chi)$} For any $\chi\in\mathcal{X}$ and $i=1,\ldots,d$ define\index{Pchii@$\bar{\mathcal{P}}^i(\chi)$}\index{Pchii@$\mathcal{P}^i(\chi)$}
\begin{align*}
 \mathcal{P}^i(\chi) &:=\{(\mathbb{P},S)\in\mathcal{P}(\chi):S_t^i=1\text{ for }t=0,\ldots,T\},\\
 \bar{\mathcal{P}}^i(\chi) &:=\{(\mathbb{P},S)\in\bar{\mathcal{P}}(\chi):S_t^i=1\text{ for }t=0,\ldots,T\}.
\end{align*}
Since $\mathcal{P}\subseteq\mathcal{P}(\chi)\subseteq\bar{\mathcal{P}}(\chi)$ and $\mathcal{K}^\ast_t$ is a cone for all $t=0,\ldots,T$, the no-arbitrage assumption implies that~$\mathcal{P}^i(\chi)$ and $\bar{\mathcal{P}}^i(\chi)$ are non-empty.

\begin{definition}[Truncated stopping time]
 Fix any $\chi\in\mathcal{X}$, $\sigma\in\mathcal{T}$. The \emph{truncated randomised stopping time} $\chi\wedge\sigma=((\chi\wedge\sigma)_t)_{t=0}^T$\index{$\cdot\wedge\cdot$}\index{chiwedgesigma@$\chi\wedge\sigma$} is defined as
\[
 (\chi\wedge\sigma)_t := \chi_t\mathbf{1}_{\{t<\sigma\}} + \chi_t^\ast\mathbf{1}_{\{t=\sigma\}}\text{ for }t=0,\ldots,T.
\]
\end{definition}
The process $\chi\wedge\sigma$ is clearly adapted and nonnegative, and moreover
\[
 \sum_{t=0}^T(\chi\wedge\sigma)_t = \sum_{t=0}^{\sigma-1}\chi_t + \chi^\ast_\sigma = 1,
\]
so it is indeed a randomised stopping time. If $\chi=\chi^\tau$ for some $\tau\in\mathcal{T}$, then clearly $(\chi\wedge\sigma)_t=\mathbf{1}_{\{\sigma\wedge\tau=t\}}$
for all $t=0,\ldots,T$, and so $\chi^\tau\wedge\sigma=\chi^{\tau\wedge\sigma}$. Denote the set of randomised stopping times truncated at $\sigma\in\mathcal{T}$ by \index{Ywedgesigma@$\mathcal{X}\wedge\cdot$}
\[
 \mathcal{X}\wedge\sigma:=\{\chi\wedge\sigma:\chi\in\mathcal{X}\}.
\]

\section{Main results and discussion}
\label{sec:main}

In this section we formally define what we mean by a game option, and present the constructions and main results.

\begin{definition}[Game option]\index{game option}\label{def:game option}
 A \emph{game option} is a derivative security that is \emph{exercised}\index{exercise} at a stopping time $\tau\in\mathcal{T}$ chosen by the buyer and \emph{cancelled}\index{cancellation} at a stopping time $\sigma\in\mathcal{T}$ chosen by the seller. At time $\sigma\wedge\tau$ the buyer receives the payoff
$
 Q_{\sigma\tau} $
from the seller, where 
\begin{equation}\label{eq:def:eta}
 Q_{st} \equiv Q^{Y,X,X'}_{st} := Y_t\mathbf{1}_{\{s>t\}} + X_s\mathbf{1}_{\{s<t\}} + X'_s\mathbf{1}_{\{s=t\}}
\end{equation}
for all $s,t=0,\ldots,T$, and $Y=(Y_t)_{t=0}^T$, $X=(X_t)_{t=0}^T$ and $X'=(X'_t)_{t=0}^T$ are adapted $\mathbb{R}^d$-valued processes such that 
\begin{align}\label{eq:ass:penalties}
 X_t-X'_t&\in\mathcal{K}_t, & X'_t-Y_t&\in\mathcal{K}_t
\end{align}
for all $t=0,\ldots,T$.
\end{definition}

In the event that the buyer exercises before the option is cancelled, i.e.~on $\{\tau<\sigma\}$, the buyer receives the payoff~$Y_\tau$ from the seller at his exercise time $\tau$. If the seller cancels the option before it is exercised, i.e.~on $\{\sigma<\tau\}$, the seller is required to deliver the payoff~$X_\sigma$ to the buyer at the cancellation time $\sigma$, which consists of $Y_\sigma$ and a penalty
\[X_\sigma-Y_\sigma=(X_\sigma-X'_\sigma) + (X'_\sigma-Y_\sigma) \in\mathcal{K}_\sigma.\]
In the event that the option is exercised and cancelled simultaneously, i.e.~on $\{\sigma=\tau\}$, the seller pays~$X'_\sigma$ to the buyer, consisting of $Y_\sigma$ and a penalty \mbox{$X'_\sigma-Y_\sigma\in\mathcal{K}_\sigma$}. The  assumptions~\eqref{eq:ass:penalties} mean that, at any time $t$, the portfolio $X_t$ payable on cancellation is at least as attractive to the buyer as the portfolio $X'_t$ payable in the event of simultaneous cancellation and exercise, which in turn is at least as attractive as the portfolio~$Y_t$ payable on exercise. It is therefore clear from Definition~\ref{def:game option} that a game option is essentially an American option with payoff process~$Y$ with the additional feature that it may be cancelled by the seller at any time (upon payment of a cancellation penalty). 

\begin{remark} \label{rem:convention}
Definition~\ref{def:game option} is slightly more general than the usual approach followed in the literature \cite[see~e.g.][]{kifer2000,Kifer2013a}, where the standard assumption is that no penalty is paid if cancellation and exercise takes place simultaneously (i.e.~$X'_t=Y_t$ for all $t$) and that no penalty is paid on maturity (i.e.~$X_T=X'_T=Y_T$). The motivation for the generalization in the present paper is that it enables elegant exploitation of the symmetry between the seller's and buyer's hedging problems; see Proposition~\ref{prop:buyer-seller-symmetry}. Nevertheless, from a practical point of view, the pricing and hedging problems depend on $X$ only through $(X_t)_{t<T}$ and on $X'$ only through $X'_T$; see the key Constructions~\ref{constr:writer} and~\ref{const:buyer} as well as Lemma~\ref{lem:hedge-equivalence}.
\end{remark}

\begin{remark} \label{rem:negative}
The property \eqref{eq:ass:penalties} imposes an ordering on the payoffs in the various scenarios for the seller and buyer, but there is no requirement in Definition~\ref{def:game option} that any of the payoffs $X_t$, $Y_t$ and $X'_t$ are solvent portfolios. The absence of such a solvency requirement makes it easy to adapt to the buyer's case, where in practice the payoffs tend to be ``negative'' in that they correspond to portfolios received rather than delivered. Typical cases are illustrated in Examples~\ref{ex:toy} and~\ref{ex:counter}.
\end{remark}

\subsection{Pricing and hedging for the seller}

A hedging strategy for the seller of a game option $(Y,X,X')$ comprises a cancellation time $\sigma$ and a self-financing trading strategy $y$ that allows the seller to the deliver the payoff without loss at any exercise time chosen by the buyer.

\begin{definition}[Hedging strategy for the seller]
 A \emph{hedging strategy for the seller} is a pair $(\sigma,y)\in\mathcal{T}\times\Phi$ satisfying
\begin{equation} \label{eq:def:seller-hedge}
 y_{\sigma\wedge\tau} - Q_{\sigma\tau}\in\mathcal{K}_{\sigma\wedge\tau} \text{ for all }\tau\in\mathcal{T}.
\end{equation}
\end{definition}

There exists at least one hedging strategy for the seller. Indeed, fixing $i=1,\ldots,d$ and defining
\[
 m:=\max\left\{\sum_{j=1}^d\pi^{ij}_t(\omega)\max\{\lvert Y^j_t(\omega)\rvert,\lvert X^j_t(\omega)\rvert,\lvert X'^j_t(\omega)\rvert\}:t=0,\ldots,T,\omega\in\Omega\right\},
\]
the (possibly expensive) buy-and-hold strategy  $y=(y_t)_{t=0}^T$ with $y_t=me^i$ for \mbox{$t=0,\ldots,T$} hedges the game option for the seller with any choice of the cancellation time $\sigma\in\mathcal{T}$.

Consider now the following construction.

\begin{construction} \label{constr:writer}
 Construct adapted set-valued mappings $(\mathcal{Y}^a_t)_{t=0}^T$, $(\mathcal{X}^a_t)_{t=0}^T$, $(\mathcal{U}^a_t)_{t=0}^T$, $(\mathcal{V}^a_t)_{t=0}^T$, $(\mathcal{W}^a_t)_{t=0}^T$, $(\mathcal{Z}^a_t)_{t=0}^T$ as follows. For all $t=0,\ldots,T$ let
\begin{align}
 \mathcal{Y}^a_t &:= Y_t + \mathcal{K}_t,&
 \mathcal{X}^a_t &:= \begin{cases}
                    X'_T + \mathcal{K}_T & \text{if } t=T,\\
                    X_t + \mathcal{K}_t & \text{if }t<T.
                   \end{cases} \label{eq:constr:A}
\end{align}
Define 
\begin{align*}
\mathcal{W}^a_T := \mathcal{V}^a_T &:= \mathcal{L}_T,&
\mathcal{Z}^a_T&:=\mathcal{X}^a_T.
\end{align*}
For $t=T-1,\ldots,0$ define by backward iteration
\begin{align}
 \mathcal{W}^a_t &:= \mathcal{Z}^a_{t+1} \cap \mathcal{L}_t, \label{eq:constr:C}\\
 \mathcal{V}^a_t &:= \mathcal{W}^a_t + \mathcal{K}_t, \label{eq:constr:B}\\
 \mathcal{Z}^a_t &:= (\mathcal{V}^a_t\cap\mathcal{Y}^a_t)\cup\mathcal{X}^a_t. \label{eq:constr:D}
\end{align}
\end{construction}

For each $t=0,\ldots,T$, the set $\mathcal{Y}^a_t$ is the collection of portfolios in $\mathcal{L}_t$ that allows the seller to settle the option in the event that the buyer exercises at time~$t$ and the seller does not cancel the option at time $t$. The set~$\mathcal{X}^a_t$ is the collection of portfolios that allows the seller to settle the option upon cancellation at time~$t$, irrespective of whether the buyer exercises at time $t$ or not. The property \eqref{eq:ass:penalties} gives that
\[
 \mathcal{X}^a_t = (X_t+\mathcal{K}_t)\cap(X'_t+\mathcal{K}_t) \text{ for }t<T.
\]
The relation $\mathcal{X}^a_T=X'_T + \mathcal{K}_T$ follows from the fact that any cancellation at the final time~$T$ must be matched by simultaneous exercise.

The following result shows that $\mathcal{Z}^a_0$ is the set of initial endowments that allow the seller to hedge the game option.

\begin{proposition}\label{prop:hedge-equivalents}
We have
\begin{equation} \label{eq:Z0=hedging}
\mathcal{Z}^a_0 = \{y_0:(\sigma,y)\text{ hedges } (Y,X,X')\text{ for the seller}\}.
\end{equation}
\end{proposition}

It is demonstrated in the proof of Proposition~\ref{prop:hedge-equivalents}, which is deferred to Section~\ref{sec:proofs}, that for each $t<T$ the sets~$\mathcal{V}^a_t$, $\mathcal{W}^a_t$ and $\mathcal{Z}^a_t$ have natural interpretations that are important to the seller of the option. The set $\mathcal{W}^a_t$ consists of those portfolios at time $t$ that allow the seller to hedge the option in the future (at time~$t+1$ or later), and $\mathcal{V}^a_t$ consists of those portfolios that may be rebalanced at time $t$ into a portfolio in~$\mathcal{W}^a_t$. The set $\mathcal{Z}^a_t$ consists of all portfolios that allow the seller to settle the option at time $t$ or any time in the future without risk of loss.

Construction \ref{constr:writer} is essentially an iteration (backwards in time) over the nodes of the price tree generated by the exchange rates; note in particular that~\eqref{eq:constr:C} could equivalently be written as
\begin{equation} \label{eq:constr:C'}
 \mathcal{W}^{a\mu}_t := \bigcap_{\nu\in\successors\mu}\mathcal{Z}^{a\nu}_{t+1} \text{ for }\mu\in\Omega_t.\tag{$\text{\ref{eq:constr:C}}'$}
\end{equation}
This property makes the construction particularly efficient for recombinant models, for which the number of nodes grow only polynomially with the number of steps in the model, despite the state space growing exponentially.

The sets $\mathcal{X}^a_t$ and $\mathcal{Y}^a_t$ in Construction \ref{constr:writer} are clearly polyhedral and non-empty for all $t=0,\ldots,T$, as are $\mathcal{V}^a_T$, $\mathcal{W}^a_T$ and $\mathcal{Z}^a_T$. The operations in Construction \ref{constr:writer} are direct addition of polyhedral cones in \eqref{eq:constr:A} and \eqref{eq:constr:B}, intersection in \eqref{eq:constr:C} and \eqref{eq:constr:D}, and union in \eqref{eq:constr:D}. The appearance of the union in \eqref{eq:constr:D} means that the sets $\mathcal{V}^a_t$, $\mathcal{W}^a_t$ and $\mathcal{Z}^a_t$ may be non-convex for some $t<T$. However, it is clear that these sets can be written as the finite union of non-empty (closed) polyhedra, and are therefore closed. In particular the closedness of $\mathcal{Z}^a_0$ is essential to Theorem~\ref{th:Z0=hedging} below.

The \emph{ask price} of the game option in terms of any currency is defined as the infimal initial endowment in that currency that would allow the seller to hedge the game option without risk. 

\begin{definition}[Ask price]\index{pai@$\pi^a_i(Y,X,X')$} \label{def:ask}
The \emph{ask price} or \emph{seller's price} or \emph{upper hedging price} of a game option $(Y,X,X')$ at time $0$ in terms of currency~$i=1,\ldots,d$ is
\begin{multline*}
 \pi^a_i(Y,X,X'):= \inf\{z\in\mathbb{R}:(\sigma,y)\in\mathcal{T}\times\Phi\text{ with }y_0=ze^i \\ \text{ hedges } (Y,X,X')\text{ for the seller}\}.
\end{multline*}
\end{definition}

The existence of the buy-and-hold strategy for the seller means that the ask price is well defined. We now present a dual representation for the ask price in terms of randomised stopping times and approximate martingale pairs. 

\begin{theorem}\label{th:ask-price}
 The ask price of a game option $(Y,X,X')$ in terms of currency $i=1,\ldots,d$ is
\begin{align*}
 \pi^a_i(Y,X,X')&= \min_{\sigma\in\mathcal{T}}\max_{\chi\in\mathcal{X}}\sup_{(\mathbb{P},S)\in\mathcal{P}^i(\chi\wedge\sigma)}\mathbb{E}_\mathbb{P}((Q_{\sigma\cdot}\cdot S_{\sigma\wedge\cdot})_\chi)\\
 &=\min_{\sigma\in\mathcal{T}}\max_{\chi\in\mathcal{X}}\max_{(\mathbb{P},S)\in\bar{\mathcal{P}}^i(\chi\wedge\sigma)}\mathbb{E}_\mathbb{P}((Q_{\sigma\cdot}\cdot S_{\sigma\wedge\cdot})_\chi),
\end{align*}
where $Q_{\sigma\cdot}\cdot S_{\sigma\wedge\cdot}$ denotes the process $(Q_{\sigma t}\cdot S_{\sigma\wedge t})_{t=0}^T$, in other words,
\[
 (Q_{\sigma\cdot}\cdot S_{\sigma\wedge\cdot})_\chi = \sum_{t=0}^{\sigma-1}\chi_tY_t\cdot S_t + \chi^\ast_{\sigma+1}X_\sigma\cdot S_\sigma + \chi_\sigma X'_\sigma\cdot S_\sigma.
\]
\end{theorem}

The proof of Theorem~\ref{th:ask-price} appears in Section~\ref{sec:proofs}. 

\begin{remark} \label{rem:expiration-date}
  \citet[Theorem~3.1]{Kifer2013a} obtained a similar dual representation for a game option $(Y,X,Y)$ in a two-currency model, that does not feature the truncated stopping time $\chi\wedge\sigma$ and stopped process $S_{\sigma\wedge\cdot}$, but rather $\chi$ and $S$. Example~\ref{ex:counter} demonstrates that these dual representations are not equivalent in general and that the representation in Theorem~\ref{th:ask-price} is indeed the correct one.
 
 The reason for the difference between the two representations can be explained intuitively in the following way. The proof of Theorem~\ref{th:ask-price} hinges on the fact that a pair $(\sigma,y)$ hedges the game option for the seller if and only if $y$ hedges an American option with payoff process $H_\sigma$ (defined in \eqref{eq:Yprime}) and random expiration date $\sigma$ for the seller. By contrast, the proof of Theorem~3.1 of \cite{Kifer2013a} claims that $(\sigma,y)$ hedges the game option $(Y,X,Y)$ for the seller if and only if $y$ hedges an American option with payoff process $\left(Q_{\sigma t}^{Y,X,Y}\right)_{t=0}^T$ and expiration date $T$ (rather than $\sigma$) for the seller. This claim does not hold true in general, because hedging such an American option would require the seller to be in a position to deliver $Q_{\sigma t}=X_\sigma$ on $\{\sigma>t\}$ at any time $t$, in other words, after the option has already been cancelled. As evidenced in Example~\ref{ex:counter}, the non-equivalence is most easily noticed when transaction costs are large at time $\sigma$ and/or $X_\sigma$ is a non-solvent portfolio.
\end{remark}

Returning to the problem of computing the ask price of a game option, the following result is a direct consequence of Proposition~\ref{prop:hedge-equivalents} and the closedness of $\mathcal{Z}^a_0$.

\begin{theorem} \label{th:Z0=hedging}
 We have
\begin{equation}  \label{eq:pi-Z0}
 \pi^a_i(Y,X,X') = \min\{x\in\mathbb{R}:xe^i\in\mathcal{Z}^a_0\}.
\end{equation}
Moreover, there exists a hedging strategy $(\hat{\sigma},\hat{y})$ for the seller such that $\hat{y}_0=\pi^a_i(Y,X,X')e^i$.
\end{theorem}

A hedging strategy $(\sigma,y)$ for the seller is called \emph{optimal} if it satisfies the properties in Theorem~\ref{th:Z0=hedging}. A procedure for constructing such a strategy can be extracted from the proof of Proposition~\ref{prop:hedge-equivalents}.

\begin{construction}\label{constr:writer:hedging}
 Construct an optimal strategy $(\hat{\sigma},\hat{y})$ for the seller as follows. Let
 \[
  \hat{y}_0:=\pi^a_i(Y,X,X')e^i.
 \]
 For each $t=0,\ldots,T-1$ and $\mu\in\Omega_t$, if $\hat{y}^\mu_t\in\mathcal{Z}^{a\mu}_t\setminus\mathcal{X}^{a\mu}_t$, then choose any
 \begin{equation} \label{eq:constr:seller:hedging}
  \hat{y}^\mu_{t+1} \in\mathcal{W}^{a\mu}_t\cap\left[\hat{y}_t^\mu-\mathcal{K}^\mu_t\right],
 \end{equation}
 otherwise put $\hat{y}^\mu_{t+1}:=\hat{y}^\mu_t$. Also define
 \[
  \hat{\sigma}:=\min\left\{t:\hat{y}_t\in\mathcal{X}^a_t\right\}.
 \]
\end{construction}

The optimal strategy for the seller is not unique in general; this is reflected in the choice \eqref{eq:constr:seller:hedging}. In practice the seller might use secondary considerations, such as a preference for holding certain currencies over others, or optimality of a secondary hedging criterion, to guide the construction of a suitable optimal hedging strategy.

Two toy examples illustrating Constructions \ref{constr:writer} and \ref{constr:writer:hedging} and Theorem~\ref{th:ask-price}, as well as a third example with a more realistic flavour can be found in Section~\ref{sec:numerical}.

\subsection{Pricing and hedging for the buyer}
\label{sec:buyer}

Consider now the hedging, pricing and optimal exercise problem for the buyer of a game option $(Y,X,X')$.

\begin{definition}[Hedging strategy for the buyer]
 A \emph{hedging strategy for the buyer} is a pair $(\tau,y)\in\mathcal{T}\times\Phi$ satisfying
\begin{equation} \label{eq:buyer-hedge:1}
 y_{\sigma\wedge\tau} + Q_{\sigma\tau}\in\mathcal{K}_{\sigma\wedge\tau} \text{ for all }\sigma\in\mathcal{T},
\end{equation}
where the payoff process $Q$ is defined in \eqref{eq:def:eta}.
\end{definition}

Observe from \eqref{eq:def:eta} that
\begin{equation} \label{eq:buyer-seller-payoff-symmetry}
 Q_{\sigma\tau} = Q^{Y,X,X'}_{\sigma\tau} = -Q^{-X,-Y,-X'}_{\tau\sigma} \text{for all } \sigma,\tau\in\mathcal{T},
\end{equation}
and moreover from \eqref{eq:ass:penalties} that for all $t=0,\ldots,T$
\begin{align*}
 -Y_t - (-X'_t) &= X'_t - Y_t \in\mathcal{K}_t,\\
 -X'_t - (-X_t) &= X_t - X'_t \in\mathcal{K}_t.
\end{align*}
Thus if $(Y,X,X')$ is the payoff of a game option, then so is $(-X,-Y,-X')$, and \eqref{eq:buyer-hedge:1} is equivalent to
\[
 y_{\tau\wedge\sigma} -Q^{-X,-Y,-X'}_{\tau\sigma} \in \mathcal{K}_{\tau\wedge\sigma} \text{ for all }\sigma\in\mathcal{T}.
\]
Thus we arrive at the following result.

\begin{proposition} \label{prop:buyer-seller-symmetry}
A pair $(\tau,y)\in\mathcal{T}\times\Phi$ hedges the game option $(Y,X,X')$ for the buyer if and only if $(\tau,y)$ hedges the game option $(-X,-Y,-X')$ for the seller. 
\end{proposition}

This symmetry means that the results and constructions developed in the previous subsection for the seller's case can also be applied to the hedging and pricing problem for the buyer, thus substantiating the claim by \citet[pp.~679--80]{Kifer2013a}. In particular, Construction \ref{constr:writer} can be applied directly, provided that $\mathcal{Y}^a_t$ and $\mathcal{X}^a_t$ in \eqref{eq:constr:A} is redefined to take into account the fact that the option is now $(-X,-Y,-X')$ rather than $(Y,X,X')$. The resulting construction reads as follows.

\begin{construction} \label{const:buyer}
 For all $t$ let
\begin{align*}
 \mathcal{Y}^b_t &:= -X_t + \mathcal{K}_t,&
 \mathcal{X}^b_t &:= \begin{cases}
                    -X'_T + \mathcal{K}_T & \text{if } t=T,\\
                    -Y_t + \mathcal{K}_t & \text{if }t<T.
                   \end{cases}
\end{align*}
Define 
\begin{align*}
\mathcal{W}^b_T := \mathcal{V}^b_T &:= \mathcal{L}_T,&
\mathcal{Z}^b_T&:=\mathcal{X}^b_T.
\end{align*}
For $t=T-1,\ldots,0$ let
\begin{align*}
 \mathcal{W}^b_t &:= \mathcal{Z}^b_{t+1} \cap \mathcal{L}_t,&
 \mathcal{V}^b_t &:= \mathcal{W}^b_t + \mathcal{K}^b_t,&
 \mathcal{Z}^b_t &:= (\mathcal{V}^b_t\cap\mathcal{Y}^b_t)\cup\mathcal{X}^b_t.
\end{align*}
\end{construction}

It follows directly from Theorem~\ref{prop:hedge-equivalents} that $\mathcal{Z}^b_0$ is the set of initial endowments that allow the buyer to hedge the option, i.e.
\[
 \mathcal{Z}^b_0=\{z:(\tau,y)\in\mathcal{T}\times\Phi\text{ with }y_0=ze^i\text{ hedges }(Y,X,X')\text{ for the buyer}\}.
\]

The bid price of a game option in any currency is the largest amount that the buyer can raise in that currency at time $0$ by using the payoff of the option as a guarantee.

\begin{definition}[Bid price]\index{pai@$\pi^b_i(Y,X,X')$} \label{def:bid}
The \emph{bid price} or \emph{lower hedging price} or \emph{buyer's price} of a game option $(Y,X,X')$ in currency $i=1,\ldots,d$ is defined as
\begin{multline*}
 \pi^b_i(Y,X,X') 
 := \sup\{-z:(\tau,y)\in\mathcal{T}\times\Phi\text{ with }y_0=ze^i \\ \text{ superhedges }(Y,X,X')\text{ for the buyer}\}.
\end{multline*}
\end{definition}

Proposition~\ref{prop:buyer-seller-symmetry} and Construction \ref{const:buyer} give that
\begin{align}
 \pi^b_i(Y,X,X') 
 &=-\pi^a_i(-X,-Y,-X')=-\inf\left\{z:ze^i\in\mathcal{Z}^b_0\right\}. \label{eq:bid-min-ask}
 \end{align} 
 A hedging strategy $(\tau,y)$ for the buyer is called \emph{optimal}\index{optimal hedging strategy} if $y_0=-\pi^b_i(Y,X,X')e^i$. Optimal hedging strategies can be generated by rewriting Construction~\ref{constr:writer:hedging} as follows.

 \begin{construction} \label{constr:buyer:hedging}
 Construct an optimal strategy $(\check{\tau},\check{y})$ for the buyer as follows. Let
 \[
  \check{y}_0:=-\pi^b_i(Y,X,X')e^i.
 \]
 For each $t=0,\ldots,T-1$ and $\mu\in\Omega_t$, if $\check{y}^\mu_t\in\mathcal{Z}^{b\mu}_t\setminus\mathcal{X}^{b\mu}_t$, then choose any
 \begin{equation} \label{eq:constr:buyer:hedging}
  \check{y}^\mu_{t+1} \in\mathcal{W}^{b\mu}_t\cap\left[\check{y}_t^\mu-\mathcal{K}^\mu_t\right],
 \end{equation}
 otherwise put $\check{y}^\mu_{t+1}:=\check{y}^\mu_t$. Also define
 \[
  \check{\tau}:=\min\left\{t:\check{y}_t\in\mathcal{X}^b_t\right\}.
 \]
\end{construction}
 
A toy example illustrating Constructions \ref{const:buyer} and \ref{constr:buyer:hedging} can be found in Section~\ref{sec:numerical}. It demonstrates that the optimal cancellation time $\hat{\sigma}$ for the seller and the optimal exercise time $\check{\tau}$ for the buyer are not the same in general, and these times may also be different from the stopping time $\hat{\sigma}\wedge\check{\tau}$ at which the option payoff is paid.

 Finally, combining Theorem~\ref{th:ask-price} with \eqref{eq:buyer-seller-payoff-symmetry} and \eqref{eq:bid-min-ask} immediately gives the following dual representation for the bid price.
 
 \begin{theorem} \label{th:bid-price}
 We have
  \begin{align*}
   \pi^b_i(Y,X,X') &= \max_{\tau\in\mathcal{T}}\min_{\chi\in\mathcal{X}}\inf_{(\mathbb{P},S)\in\mathcal{P}^i(\chi\wedge\tau)}\mathbb{E}_\mathbb{P}((Q_{\cdot\tau}\cdot S_{\cdot\wedge\tau})_\chi) \\&= \max_{\tau\in\mathcal{T}}\min_{\chi\in\mathcal{X}}\min_{(\mathbb{P},S)\in\bar{\mathcal{P}}^i(\chi\wedge\tau)}\mathbb{E}_\mathbb{P}((Q_{\cdot\tau}\cdot S_{\cdot\wedge\tau})_\chi),
  \end{align*}
  where $Q_{\cdot\tau}\cdot S_{\cdot\wedge\tau}$ denotes the process $(Q^{Y,X,X'}_{s \tau}\cdot S_{s\wedge\tau})_{s=0}^T$, i.e.
  \[
   (Q_{\cdot\tau}\cdot S_{\cdot\wedge\tau})_\chi = \sum_{s=0}^{\tau-1}\chi_sX_s\cdot S_s + \chi^\ast_{\tau+1}Y_\tau\cdot S_\tau + \chi_\tau X'_\tau\cdot S_\tau.
  \]
  \end{theorem}
  
  The representation in Theorem~\ref{th:bid-price} is different from the representation by \citet[Theorem~3.1]{Kifer2013a} for a game option $(Y,X,X')$ in a two-currency model, for reasons already discussed in the context of Theorem~\ref{th:ask-price}; see Remark~\ref{rem:expiration-date}.

\section{Proofs of results for the seller's case}
\label{sec:proofs}

\begin{proof}[Proof of Proposition~\ref{prop:hedge-equivalents}]
	Fix any $z\in\mathcal{Z}^a_0$. We claim that there exists a hedging strategy $(\sigma,y)$ for the seller with $y_0=z$. To this end, we construct $y=(y_t)_{t=0}^T$ together with a non-decreasing sequence $(\sigma_t)_{t=0}^T$ of stopping times. Define
	\begin{align*}
		y_0 &:= z, & \sigma_0:=
		\begin{cases}
			0 &\text{if }z\in\mathcal{X}^a_0,\\
			1 &\text{if }z\in\mathcal{Z}^a_0\setminus\mathcal{X}^a_0.
		\end{cases}
	\end{align*}
	If $\sigma_0=0$, let $y_1:=y_0$. If $\sigma_0=1$, then
	\[
		y_0\in\mathcal{Z}^a_0\setminus\mathcal{X}^a_0\subseteq\mathcal{V}^a_0\cap\mathcal{Y}^a_0\subseteq\mathcal{W}^a_0+\mathcal{K}_0,
	\]
	so that there exists some $y_1\in\mathcal{W}^a_0$ such that $y_0-y_1\in\mathcal{K}_0$ and $y_1\in\mathcal{Z}^a_1$.

	Suppose by induction that for some $t>0$ we have constructed $y_0,\ldots,y_t$ and non-decreasing $\sigma_0,\ldots,\sigma_{t-1}$ such that for $s=0,\ldots,t-1$ we have $y_{s+1}$ being $\mathcal{F}_s$-measurable, $\sigma_s\le s+1$, $y_s-y_{s+1}\in\mathcal{K}_s$, $y_{\sigma_s}\in\mathcal{X}^a_{\sigma_s}$ on $\{\sigma_s\le s\}$ and $y_u\in\mathcal{Z}^a_u\setminus\mathcal{X}^a_u$ on $\{\sigma_s>u\}$ for all $u=0,\ldots,s$. 	Define
	\[
	 \sigma_t := \sigma_{t-1}\mathbf{1}_{\{\sigma_{t-1}<t\}} + t\mathbf{1}_{\{\sigma_{t-1}=t\}\cap\{y_t\in\mathcal{X}^a_t\}} +(t+1)\mathbf{1}_{\{\sigma_{t-1}=t\}\cap\{y_t\in\mathcal{Z}^a_t\setminus\mathcal{X}^a_t\}}.
	\]
	On the set \[\{\sigma_{t}>t\} =\{\sigma_{t}=t+1\} =\{\sigma_{t-1}=t\}\cap\{y_t\in\mathcal{Z}^a_t\setminus\mathcal{X}^a_t\}\] we have
	\[y_t\in\mathcal{Z}^a_t\setminus\mathcal{X}^a_t\subseteq\mathcal{V}^a_t\cap\mathcal{Y}^a_t\subseteq\mathcal{V}^a_t=\mathcal{W}^a_t+\mathcal{K}_t,\]
	so there exists an $\mathcal{F}_t$-measurable random variable $x\in\mathcal{Z}^a_{t+1}$ such that $y_t-x\in\mathcal{K}_t$ on this set.
	Now define the $\mathcal{F}_t$-measurable random variable
	\[
	 y_{t+1} := x\mathbf{1}_{\{\sigma_t=t+1\}} + y_t\mathbf{1}_{\{\sigma_t\le t\}}.
	\]
	Since $0\in\mathcal{K}_t$ we have $y_t-y_{t+1}\in\mathcal{K}_t$. Moreover,
	\begin{align*}
	 y_{\sigma_t}
	&= y_{\sigma_t}\mathbf{1}_{\{\sigma_{t}<t\}} + y_t\mathbf{1}_{\{\sigma_{t}=t\}} + y_{t+1}\mathbf{1}_{\{\sigma_{t}=t+1\}}\\
	&= y_{\sigma_{t-1}}\mathbf{1}_{\{\sigma_t<t\}} + y_t\mathbf{1}_{\{\sigma_{t-1}=t\}\cap\{y_t\in\mathcal{X}^a_t\}}  + y_{t+1}\mathbf{1}_{\{\sigma_{t}=t+1\}}
	\end{align*}
	so $y_{\sigma_t}\in\mathcal{X}^a_{\sigma_t}$ on the set
	\[
	  \{\sigma_t\le t\}=\{\sigma_{t-1}<t\}\cup[\{\sigma_{t-1}=t\}\cap\{y_t\in\mathcal{X}^a_t\}].
	\]
	This concludes the inductive step. 
	
	Let $\sigma:=\sigma_T$; then the pair $(\sigma,y)\in\mathcal{T}\times\Phi$ satisfies 
	\begin{align*}
	 y_\sigma&\in\mathcal{X}^a_\sigma, & y_t&\in\mathcal{Z}^a_t\setminus\mathcal{X}^a_t\text{ on }\{t<\sigma\}\text{ for }t=0,\ldots,T-1.
	\end{align*}
        Fix any stopping time $\tau$. On $\{\tau<\sigma\}$ we have \[y_\tau\in\mathcal{Z}^a_\tau\setminus\mathcal{X}^a_\tau\subseteq\mathcal{V}^a_\tau\cap\mathcal{Y}_\tau\subseteq\mathcal{Y}_\tau=Y_\tau+\mathcal{K}_\tau.\]
	On the set $\{\tau=\sigma\}$ we have $y_\sigma\in\mathcal{X}^a_\sigma\subseteq X'_\sigma+\mathcal{K}_\sigma$. On the set $\{\tau>\sigma\}$ we have $\sigma<\tau\le T$ and $y_\sigma\in\mathcal{X}^a_\sigma=X_\sigma+\mathcal{K}_\sigma$. Thus $y_{\sigma\wedge\tau} - Q_{\sigma\tau}\in\mathcal{K}_{\sigma\wedge\tau}$, from which it follows that $(\sigma,y)$ hedges the game option for the seller.
	
	Conversely, suppose that $(\sigma,y)$ hedges the game option for the seller. We show by backward induction that $y_t\in\mathcal{Z}^a_t$ on $\{t\le\sigma\}$ and $y_t\in\mathcal{V}^a_t\cap\mathcal{Y}^a_t$ on $\{t<\sigma\}$ for all $t=0,\ldots,T$, from which it can be deduced that $z=y_0\in\mathcal{Z}^a_0$, which completes the proof. At time $T$ we have $\{T=\sigma\}=\{T\le\sigma\}$, so clearly $y_T\in X'_T+\mathcal{K}_T=\mathcal{X}^a_T=\mathcal{Z}^a_T$ on $\{T\le\sigma\}$.

	For any $t<T$, suppose that $y_{t+1}\in\mathcal{Z}^a_{t+1}$ on $\{t+1\le\sigma\}=\{t<\sigma\}$. This means that $y_{t+1}\in\mathcal{W}^a_t$ on $\{t<\sigma\}$ as $y_{t+1}\in\mathcal{L}_t$. Moreover $y_t-y_{t+1}\in\mathcal{K}_t$ implies that $y_t\in\mathcal{V}^a_t$ on $\{t<\sigma\}$, and therefore
	\[
	 y_t\in\mathcal{V}^a_t\cap[Y_t+\mathcal{K}_t]=\mathcal{V}^a_t\cap\mathcal{Y}^a_t\subseteq\mathcal{Z}^a_t\text{ on }\{t<\sigma\}.
	\]
	On the set $\{t=\sigma\}=\{t\le\sigma\}\setminus\{t<\sigma\}$ we have $y_t\in X_t+\mathcal{K}_t=\mathcal{X}^a_t\subseteq\mathcal{Z}^a_t$. This concludes the induction.
\end{proof}

The next result will play an important role in the proof of Theorem~\ref{th:ask-price}. Define the auxiliary process
\begin{align}
 H_{st} \equiv H^{(Y,X,X')}_{st}
 &:= Y_t\mathbf{1}_{\{s>t\}} + X_s\mathbf{1}_{\{s=t<T\}} + X'_s\mathbf{1}_{\{s=t=T\}} \label{eq:Yprime}
\end{align}
for all $s,t=0,\ldots,T$. Observe that the process $H_{\sigma}=(H_{\sigma t})_{t=0}^T$ is adapted for any $\sigma\in\mathcal{T}$, and that $H_{\sigma t} = Q_{\sigma t} = Y_t$ on $\{\sigma>t\}$ and $H_{\sigma t} = 0$ on $\{\sigma<t\}$ for all $t=0,\ldots,T$. 

The payoff $H_{\sigma t}$ can be interpreted as the payoff that a seller with pre-selected cancellation time $\sigma$ needs to be prepared to deliver at time $t$ if it is known that the buyer will not exercise at the same time as when the option is cancelled (except at time $t=T$); this is effectively the worst case scenario for such a seller because of \eqref{eq:ass:penalties}.

\begin{lemma}\label{lem:hedge-equivalence}
 A pair $(\sigma,y)\in\mathcal{T}\times\Phi$ hedges the game option for the seller if and only if
 \begin{equation}\label{eq:prop:hedge-equivalence}
  y_{\sigma\wedge \tau} - H_{\sigma \tau} \in\mathcal{K}_{\sigma\wedge \tau} \text{ for all }\tau\in\mathcal{T}.
 \end{equation}
\end{lemma}

\begin{proof} Throughout the proof we shall make frequent and implicit use of the fact that $\mathcal{K}_t$ is a pointed cone for all $t$, and in particular the property $0\in\mathcal{K}_t$.

Fix any $(\sigma,y)\in\mathcal{T}\times\Phi$ satisfying \eqref{eq:prop:hedge-equivalence}. In view of \eqref{eq:Yprime}, this implies that
 \begin{align}
  (y_{\sigma\wedge \tau} - Y_{\sigma\wedge \tau})\mathbf{1}_{\{\sigma>\tau\}} & \in\mathcal{K}_{\sigma\wedge\tau}  \text{ for all } \tau\in\mathcal{T},\label{eq:A}\\
  (y_{\sigma} - X_{\sigma})\mathbf{1}_{\{\sigma=\tau<T\}} & \in\mathcal{K}_{\sigma\wedge\tau}  \text{ for all } \tau\in\mathcal{T}, \label{eq:B'}\\
  (y_{\sigma} - X'_{\sigma})\mathbf{1}_{\{\sigma=\tau=T\}} & \in\mathcal{K}_{\sigma\wedge\tau}.  \text{ for all } \tau\in\mathcal{T}, \label{eq:C'}
 \end{align}
 Substituting $\tau=0,\ldots,T-1$ into \eqref{eq:B'} and $\tau=T$ into \eqref{eq:C'} gives
 \begin{align}
  (y_{\sigma} - X_{\sigma})\mathbf{1}_{\{\sigma<T\}} & \in\mathcal{K}_{\sigma}, \label{eq:B}\\
  (y_{\sigma} - X'_{\sigma})\mathbf{1}_{\{\sigma=T\}} & \in\mathcal{K}_{\sigma}. \label{eq:C}
  \end{align}
Property \eqref{eq:ass:penalties} together with \eqref{eq:B}--\eqref{eq:C} then leads to $y_\sigma-X'_\sigma\in\mathcal{K}_\sigma$, whence
\begin{multline}
  (y_{\sigma\wedge\tau} - X'_{\sigma\wedge\tau})\mathbf{1}_{\{\sigma=\tau\}} = (y_{\sigma} - X'_{\sigma})\mathbf{1}_{\{\sigma=\tau\}}  \\\in \mathbf{1}_{\{\sigma=\tau\}}\mathcal{K}_{\sigma} =\mathbf{1}_{\{\sigma=\tau\}}\mathcal{K}_{\sigma\wedge\tau} \subseteq\mathcal{K}_{\sigma\wedge\tau}  \text{ for all } \tau\in\mathcal{T}. \label{eq:D}
  \end{multline}
  It also follows from \eqref{eq:B} that
  \begin{multline}
  (y_{\sigma\wedge\tau} - X_{\sigma\wedge\tau})\mathbf{1}_{\{\sigma<\tau\}} 
  = (y_{\sigma} - X_{\sigma})\mathbf{1}_{\{\sigma<\tau\}}
  = (y_{\sigma} - X_{\sigma})\mathbf{1}_{\{\sigma<T\}}\mathbf{1}_{\{\sigma<\tau\}}
  \\\in \mathbf{1}_{\{\sigma<\tau\}}\mathcal{K}_{\sigma} = \mathbf{1}_{\{\sigma<\tau\}}\mathcal{K}_{\sigma\wedge\tau} \subseteq
  \mathcal{K}_{\sigma\wedge\tau}\text{ for all } \tau\in\mathcal{T}. \label{eq:E}
 \end{multline}
Properties \eqref{eq:A}, \eqref{eq:D} and \eqref{eq:E} lead to \eqref{eq:def:seller-hedge}, and so $(\sigma,y)$ hedges the game option for the seller.

Suppose conversely that $(\sigma,y)\in\mathcal{T}\times\Phi$ hedges the game option for the seller. Property \eqref{eq:def:seller-hedge} gives~\eqref{eq:A} and (upon choosing $\tau=T$)
\[
  y_{\sigma} - X_{\sigma}\mathbf{1}_{\{\sigma<T\}} - X'_{\sigma}\mathbf{1}_{\{\sigma=T\}} \in\mathcal{K}_{\sigma},
\]
from which it follows that
\begin{multline}
    (y_{\sigma\wedge\tau} - X_{\sigma})\mathbf{1}_{\{\sigma=\tau<T\}} + (y_{\sigma\wedge\tau} - X'_{\sigma})\mathbf{1}_{\{\sigma=\tau=T\}} \\
    \begin{aligned}
    &= y_{\sigma\wedge\tau}\mathbf{1}_{\{\sigma=\tau\}} - X_{\sigma}\mathbf{1}_{\{\sigma=\tau<T\}} - X'_{\sigma}\mathbf{1}_{\{\sigma=\tau=T\}} \\
    &= \left(y_{\sigma\wedge\tau} - X_{\sigma}\mathbf{1}_{\{\sigma<T\}} - X'_{\sigma}\mathbf{1}_{\{\sigma=T\}}\right)\mathbf{1}_{\{\sigma=\tau\}} \\
   &= \left(y_{\sigma} - X_{\sigma}\mathbf{1}_{\{\sigma<T\}} - X'_{\sigma}\mathbf{1}_{\{\sigma=T\}}\right)\mathbf{1}_{\{\sigma=\tau\}}
   \end{aligned}\\
   \in \mathbf{1}_{\{\sigma=\tau\}}\mathcal{K}_{\sigma} = \mathbf{1}_{\{\sigma=\tau\}}\mathcal{K}_{\sigma\wedge\tau}\subseteq \mathcal{K}_{\sigma\wedge\tau} \text{ for all } \tau\in\mathcal{T}.\label{eq:A*}
\end{multline}
Properties \eqref{eq:A} and \eqref{eq:A*} together give \eqref{eq:prop:hedge-equivalence}, which completes the proof.
\end{proof}

\begin{proof}[Proof of Theorem~\ref{th:ask-price}]
Lemma \ref{lem:hedge-equivalence} shows that, for $\sigma\in\mathcal{T}$ given, the pair $(\sigma,y)\in\mathcal{T}\times\Phi$ hedges the game option $(Y,X,X')$ for the seller if and only if \eqref{eq:prop:hedge-equivalence} holds, equivalently
\[
 y_\tau - H_{\sigma\tau} \in\mathcal{K}_\tau \text{ for all }\tau\in\mathcal{T}\text{ such that } \tau\le\sigma.
\]
 Definition~\ref{def:ask} and the finiteness of $\mathcal{T}$ then give that
 \begin{multline*}
  \pi^a_i(Y,X,X') \\
  \begin{aligned}
  &= \min_{\sigma\in\mathcal{T}}\inf\{z\in\mathbb{R}:(\sigma,y)\text{ hedges } (Y,X,X')\text{ for the seller and }y_0=ze^i\} \\
  &= \min_{\sigma\in\mathcal{T}}p^a_i(\sigma),
  \end{aligned}
 \end{multline*}
 where\index{pai@$p^a_i(\sigma)$}
 \begin{multline}
  p^a_i(\sigma):=\inf\{z\in\mathbb{R}:y\in\Phi\text{ such that } \\ y_0=ze^i,y_\tau - H_{\sigma\tau} \in\mathcal{K}_\tau \text{ for all }\tau\in\mathcal{T},\tau\le\sigma\}.\label{eq:def:pai}
 \end{multline}

 This means that any pair $(\sigma,y)$ hedges the game option $(Y,X,X')$ for the seller if and only if, in the terminology of \citet{Roux_Zastawniak2015}, the strategy $y$ superhedges an option with payoff process $H_{\sigma}=(H_{\sigma t})_{t=0}^T$ that can be exercised by the buyer at any stopping time $\tau$ satisfying
 \[
  \{\tau=t\}\subseteq\mathcal{E}_t:=\{t\le\sigma\}\text{ for all } t=0,\ldots,T
 \]
 for the seller. Intuitively, this is an American option with (random) expiration date $\sigma$. The quantity $p^a_i(\sigma)$ in \eqref{eq:def:pai} is the ask price of such an option in asset~$i$, and \citet[Theorem~3]{Roux_Zastawniak2015} established that
\[
 p^a_i(\sigma) = \max_{\chi\in\mathcal{X}^\mathcal{E}}\sup_{(\mathbb{P},S)\in\mathcal{P}^i(\chi)}\mathbb{E}_\mathbb{P}((H_\sigma\cdot S)_\chi)= \max_{\chi\in\mathcal{X}^\mathcal{E}}\max_{(\mathbb{P},S)\in\bar{\mathcal{P}}^i(\chi)}\mathbb{E}_\mathbb{P}((H_\sigma\cdot S)_\chi)
\]
where $H_\sigma\cdot S=(H_{\sigma t}\cdot S_t)_{t=0}^T$ and
\[
 \mathcal{X}^\mathcal{E} := \{\chi\in\mathcal{X}:\{\chi_t>0\}\subseteq\mathcal{E}_t\text{ for all }t=0,\ldots,T\}=\mathcal{X}\wedge\sigma.
\]
It then follows that
\begin{equation}
 p^a_i(\sigma)=\max_{\chi\in\mathcal{X}}\sup_{(\mathbb{P},S)\in\mathcal{P}^i(\chi\wedge\sigma)}\mathbb{E}_\mathbb{P}((H_\sigma\cdot S)_{\chi\wedge\sigma})=\max_{\chi\in\mathcal{X}}\max_{(\mathbb{P},S)\in\bar{\mathcal{P}}^i(\chi\wedge\sigma)}\mathbb{E}_\mathbb{P}((H_\sigma\cdot S)_{\chi\wedge\sigma}),
\end{equation}
so that
\begin{align}
 \pi^a_i(Y,X,X') &= \min_{\sigma\in\mathcal{T}}\max_{\chi\in\mathcal{X}}\sup_{(\mathbb{P},S)\in\mathcal{P}^i(\chi\wedge\sigma)}\mathbb{E}_\mathbb{P}((H_\sigma\cdot S)_{\chi\wedge\sigma})\nonumber\\
 &=\min_{\sigma\in\mathcal{T}}\max_{\chi\in\mathcal{X}}\max_{(\mathbb{P},S)\in\bar{\mathcal{P}}^i(\chi\wedge\sigma)}\mathbb{E}_\mathbb{P}((H_\sigma\cdot S)_{\chi\wedge\sigma}). \label{eq:p-max-eta-2}
\end{align}

Fix now any $\sigma\in\mathcal{T}$, $\chi\in\mathcal{X}$ and $(\mathbb{P},S)\in\bar{\mathcal{P}}(\chi\wedge\sigma)$ and note that
\begin{align}
 &(H_\sigma\cdot S)_{\chi\wedge\sigma} \nonumber\\
 &= (Q_{\sigma\cdot}\cdot S_{\sigma\wedge\cdot})_\chi + (\chi_\sigma^\ast\mathbf{1}_{\{\sigma<T\}}-\chi^\ast_{\sigma+1})X_\sigma\cdot S_\sigma + (\chi^\ast_\sigma\mathbf{1}_{\{\sigma=T\}} - \chi_\sigma)X'_\sigma\cdot S_\sigma\nonumber\\
 &= (Q_{\sigma\cdot}\cdot S_{\sigma\wedge\cdot})_\chi + \left(\chi_\sigma\mathbf{1}_{\{\sigma<T\}}-\chi^\ast_{\sigma+1}\mathbf{1}_{\{\sigma=T\}}\right)X_\sigma\cdot S_\sigma - \chi_\sigma\mathbf{1}_{\{\sigma<T\}}X'_\sigma\cdot S_\sigma\nonumber\\
 &= (Q_{\sigma\cdot}\cdot S_{\sigma\wedge\cdot})_\chi + \chi_\sigma\mathbf{1}_{\{\sigma<T\}}(X_\sigma-X'_\sigma)\cdot S_\sigma. \label{eq:gamma-eta}
\end{align}
This follows from the properties of $\chi^\ast$: in particular $\chi^\ast_{T+1}=0$, $\chi^\ast_T=\chi_T$ and $\chi^\ast_\sigma = \chi^\ast_{\sigma+1}+\chi_\sigma$. Since $X_\sigma-X'_\sigma\in\mathcal{K}_\sigma$ by \eqref{eq:ass:penalties} and $S_\sigma\in\mathcal{K}^\ast_\sigma$, we immediately have
\begin{equation} \label{eq:gamma-eta:1}
 (H_\sigma\cdot S)_{\chi\wedge\sigma} \ge (Q_{\sigma\cdot}\cdot S_{\sigma\wedge\cdot})_\chi.
\end{equation}
Define the stopping time $\chi'=(\chi'_t)\in\mathcal{X}$ by
\[
 \chi'_t := \chi_t\mathbf{1}_{\{t<\sigma\}} +\chi^{\ast}_\sigma\mathbf{1}_{\{t=T\}} \text{ for } t=0,\ldots,T;
\]
then $\chi\wedge\sigma=\chi'\wedge\sigma$ and so $(\mathbb{P},S)\in\bar{\mathcal{P}}(\chi'\wedge\sigma)$. Moreover, since $\chi'_\sigma\mathbf{1}_{\{\sigma<T\}}=0$ it follows from \eqref{eq:gamma-eta} that
\begin{equation} \label{eq:gamma-eta:2}
  (H_\sigma\cdot S)_{\chi\wedge\sigma} = (H_\sigma\cdot S)_{\chi'\wedge\sigma} = (Q_{\sigma\cdot}\cdot S_{\sigma\wedge\cdot})_{\chi'}.
\end{equation}
Combining \eqref{eq:gamma-eta:1} and \eqref{eq:gamma-eta:2} then gives
\[
 \max_{\chi\in\mathcal{X}}\max_{(\mathbb{P},S)\in\bar{\mathcal{P}}^i(\chi\wedge\sigma)}\mathbb{E}_\mathbb{P}((H_\sigma\cdot S)_{\chi\wedge\sigma}) = \max_{\chi\in\mathcal{X}}\max_{(\mathbb{P},S)\in\bar{\mathcal{P}}^i(\chi\wedge\sigma)}\mathbb{E}_\mathbb{P}((Q_{\sigma\cdot}\cdot S_{\sigma\wedge\cdot})_\chi),
\]
and the result follows from \eqref{eq:p-max-eta-2}.
\end{proof}

\section{Numerical examples}
\label{sec:numerical}

Three numerical examples are presented in this section. The first is a toy example to illustrate the constructions in Section~\ref{sec:main}. The second illustrates Theorem~\ref{th:ask-price} and serves as a minimal counterexample to Theorem~3.1 of \citet{Kifer2013a}. The final example has a more realistic flavour.

\begin{example}\label{ex:toy}
 A game option $(Y,X,X')$ in a binary two-step two-currency model is presented in Figure \ref{fig:two-step-model}. The model is recombinant and has transaction costs only at node $\mathrm{u}$ at time $1$, and the option is path-independent and has no cancellation penalties at time $2$.
 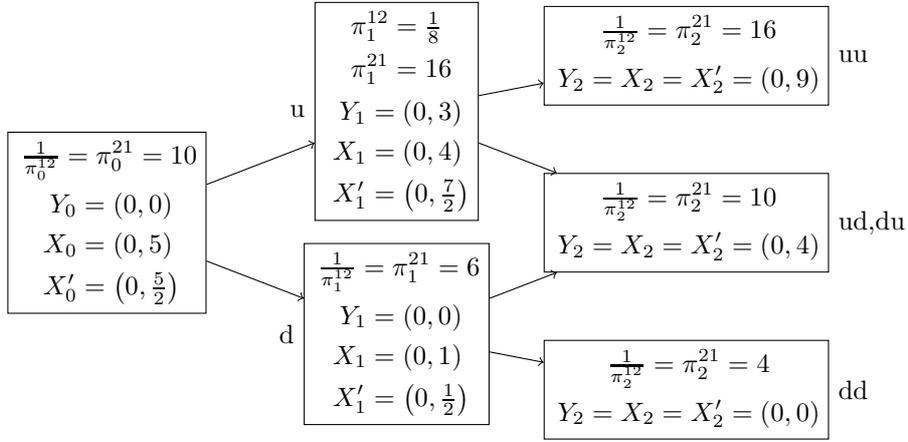
\begin{figure}
 \begin{center}
 \begin{tikzpicture}[model,twostep, max label lines = 2.5]
  \node (root) {$\begin{gathered}\tfrac{1}{\pi^{12}_0} = \pi^{21}_0 = 10\\
	\begin{aligned}Y_0&=\twovector{0}{0} \\ X_0&=\twovector{0}{5} \\ X'_0&=\twovector{0}{\tfrac{5}{2}}\end{aligned}\end{gathered}$}
    child { node[label=left:u] (u) {$\begin{gathered}\begin{aligned}\pi^{12}_1 &=\tfrac{1}{8} \\ \pi^{21}_1 &= 16\end{aligned}\\
	\begin{aligned}Y_1&=\twovector{0}{3} \\ X_1&=\twovector{0}{4} \\ X'_1&=\twovector{0}{\tfrac{7}{2}}\end{aligned}\end{gathered}$}
    	child { node[label=right:uu] (uu) {$\begin{gathered}\tfrac{1}{\pi^{12}_2} =\pi^{21}_2 = 16\\
	Y_2=X_2=X'_2=\twovector{0}{9}\end{gathered}$}}
    	child { node [draw=none] {} edge from parent [draw=none]}
	}
    child { node[label=left:d] (d) {$\begin{gathered}\tfrac{1}{\pi^{12}_1} =\pi^{21}_1 = 6\\
	\begin{aligned}Y_1&=\twovector{0}{0} \\ X_1&=\twovector{0}{1} \\ X'_1&=\twovector{0}{\tfrac{1}{2}}\end{aligned}\end{gathered}$}
    	child { node [draw=none] {} edge from parent [draw=none]}
    	child { node[label=right:dd] (dd) {$\begin{gathered}\tfrac{1}{\pi^{12}_2} = \pi^{21}_2 = 4\\
	Y_2=X_2=X'_2=\twovector{0}{0}\end{gathered}$}}
	};

	
	\node[label=right:{ud,du}] (ud) at (2\leveldistance,0) {$\begin{gathered}\tfrac{1}{\pi^{12}_2} = \pi^{21}_2 = 10\\
	Y_2=X_2=X'_2=\twovector{0}{4}\end{gathered}$};
	
    \draw (u) -- (ud);
    \draw (d) -- (ud);
\end{tikzpicture}
\end{center}
  \caption{Game option in binary two-step two-currency model, Example \ref{ex:toy}}
  \label{fig:two-step-model}
 \end{figure}
 
 Let us use Construction \ref{constr:writer} to find the set of hedging endowments for the seller. Clearly
 \begin{align*}
  \mathcal{Z}^{a\mathrm{uu}}_2 &= \left\{\twovector{x^1}{x^2}\in\mathbb{R}^2:16x^1 + x^2 \ge 9\right\},\\
  \mathcal{Z}^{a\mathrm{ud}}_2 &= \left\{\twovector{x^1}{x^2}\in\mathbb{R}^2:10x^1 + x^2 \ge 4\right\},\\
  \mathcal{Z}^{a\mathrm{dd}}_2 &= \left\{\twovector{x^1}{x^2}\in\mathbb{R}^2:4x^1 + x^2 \ge 0\right\}
 \end{align*}
 at time $t=2$. 
 
 For time $t=1$, consider the node $\mathrm{u}$. We obtain $\mathcal{W}^{a\mathrm{u}}_1$ from \eqref{eq:constr:C'} and $\mathcal{V}^{a\mathrm{u}}_1$ from \eqref{eq:constr:B}; both procedures are shown graphically in Figure \ref{fig:ex:toy:Zau1}. Observe that the magnitude of the transaction costs at this node means that $\mathcal{W}^{a\mathrm{u}}_1+\mathcal{K}^\mathrm{u}_1=\mathcal{W}^{a\mathrm{u}}_1$, whence $\mathcal{V}^{a\mathrm{u}}_1=\mathcal{W}^{a\mathrm{u}}_1$. The next step is to compute $\mathcal{Z}^{a\mathrm{u}}_1$ from \eqref{eq:constr:D}. As shown in Figure \ref{fig:ex:toy:Zau1}, this can be done by finding the intersection of $\mathcal{V}^{a\mathrm{u}}_1$ and $\mathcal{Y}^{a\mathrm{u}}_1$ first and then taking the union with $\mathcal{X}^{a\mathrm{u}}_1$. The non-convexity of $\mathcal{Z}^{a\mathrm{u}}_1$ is due to the magnitude of the transaction costs and the shape of the payoff at this node. Similar considerations at the node $\mathrm{d}$ give that
 \begin{equation} \label{eq:ex:Zad1}
  \mathcal{Z}^{a\mathrm{d}}_1 = \mathcal{X}^{a\mathrm{d}}_1 = \left\{\twovector{x^1}{x^2}:6x^1+x^2\ge1\right\}.
 \end{equation}

   \begin{figure}
  \newcommand{\minx}{-0.1}
  \newcommand{\maxx}{1.1}
  \newcommand{\miny}{-7}
  \newcommand{\maxy}{10.0}
   \centering
   \begin{tabular}{@{}r@{}r@{}}
    \footnotesize\begin{tikzpicture}[x=0.5/(\maxx-(\minx))*\figurewidth,y=0.5/(\maxy-(\miny))*\figureheight]
   \draw[draw=none,pattern=north east lines, pattern color=lightgray] ({0 - (\maxy - 9)/16} ,\maxy) -- (0, 9) -- ({0-(\miny-9)/16},\miny) -- (\maxx,\miny) -- (\maxx,\maxy) -- (\minx,\maxy) -- cycle; 
    \draw[draw=none,pattern=north west lines, pattern color=lightgray] (\minx ,{4-(\minx-0)*10}) -- (0, 4) -- ({0-(\miny-4)/10},\miny) -- (\maxx,\miny) -- (\maxx,\maxy) -- (\minx,\maxy) -- cycle; 
  \draw[draw=none,fill=lightgray,opacity=0.5] ({5/6 - (\maxy - (-13/3))/16} ,\maxy) -- (5/6, -13/3) -- (\maxx,{-13/3 -(\maxx-5/6)*10}) -- (\maxx,\maxy) -- cycle; 
   
   \draw[thick] ({5/6 - (\maxy - (-13/3))/16} ,\maxy) -- (5/6, -13/3) -- (\maxx,{-13/3 -(\maxx-5/6)*10});
   \draw (\maxx,\maxy) node[below left] {$\mathcal{W}^{a\mathrm{u}}_1=\mathcal{Z}^{a\mathrm{uu}}_2\cap\mathcal{Z}^{a\mathrm{ud}}_2$};
   
    \draw ({0 - (\maxy - 9)/16} ,\maxy) -- (0, 9) -- ({0-(\miny-9)/16},\miny) node[sloped,above,pos=0.3] {$\mathcal{Z}^{a\mathrm{uu}}_2$};
     \draw (\minx ,{4-(\minx-0)*10}) -- (0, 4) -- ({0-(\miny-4)/10},\miny) node[sloped,above,pos=0.25] {$\mathcal{Z}^{a\mathrm{ud}}_2$};

     \foreach \y/\yname in {{-13/3}/-\tfrac{13}{3},9,4}
        \draw (0,\y) node[left] {\footnotesize$\yname$} -- ++(1ex,0);    
 
      \foreach \x/\xname in {{5/6}/\tfrac{5}{6},{2/5}/\tfrac{2}{5}}
        \draw (\x,0) node[below] {\footnotesize$\xname$} -- ++(0,1ex);         

      \foreach \x/\xname in {{9/16}/\tfrac{9}{16}}
        \draw (\x,0) -- ++(0,1ex) node[above] {\footnotesize$\xname$};
        
    \draw (\minx,0) -- (\maxx,0) node[above left] {$x^1$};
    \draw (0,\maxy) -- (0,\miny) node[above right] {$x^2$};
 \end{tikzpicture}
 &
    \footnotesize\begin{tikzpicture}[x=0.5/(\maxx-(\minx))*\figurewidth,y=0.5/(\maxy-(\miny))*\figureheight]
   \draw[draw=none,pattern=north east lines, pattern color=lightgray] ({5/6 - (\maxy - (-13/3))/16} ,\maxy) -- (5/6, -13/3) -- (\maxx,{-13/3 -(\maxx-5/6)*10}) -- (\maxx,\maxy) -- cycle; 
   \draw[draw=none,fill=lightgray,opacity=0.5] ({5/6 - (\maxy - (-13/3))/16} ,\maxy) -- (5/6, -13/3) -- (\maxx,{-13/3 -(\maxx-5/6)*10}) -- (\maxx,\maxy) -- cycle; 
   
   \draw[thick] ({5/6 - (\maxy - (-13/3))/16} ,\maxy) -- (5/6, -13/3) -- (\maxx,{-13/3 -(\maxx-5/6)*10});
   \draw (\maxx,\maxy) node[below left] {$\mathcal{V}^{a\mathrm{u}}_1=\mathcal{W}^{a\mathrm{u}}_1$};
   
   \draw[dashed] ({5/6 - (\maxy - (-13/3))/16} ,\maxy)  -- (5/6, -13/3) -- (\maxx,{-13/3 -(\maxx-5/6)*8}); 
   \draw[dashed] (\minx, {0 - (\minx-0)*16}) -- (0, 0) -- ({0 - (\miny-0)/8},\miny)  node[sloped, above, pos=0.5] {lower boundary of $\mathcal{K}^\mathrm{u}_1$};
   
     \foreach \y/\yname in {{-13/3}/-\tfrac{13}{3},9}
        \draw (0,\y) node[left] {\footnotesize$\yname$} -- ++(1ex,0);    
 
      \foreach \x/\xname in {{9/16}/\tfrac{9}{16},{5/6}/\tfrac{5}{6}}
        \draw (\x,0) node[below] {\footnotesize$\xname$} -- ++(0,1ex);         

    \draw (\minx,0) -- (\maxx,0) node[above left] {$x^1$};
    \draw (0,\maxy) -- (0,\miny) node[above right] {$x^2$};
 \end{tikzpicture}
 \\
    \footnotesize\begin{tikzpicture}[x=0.5/(\maxx-(\minx))*\figurewidth,y=0.5/(\maxy-(\miny))*\figureheight]
   \draw[draw=none,pattern=north east lines, pattern color=lightgray]  ({5/6 - (\maxy - (-13/3))/16} ,\maxy) -- (5/6, -13/3) -- (\maxx,{-13/3 -(\maxx-5/6)*10}) -- (\maxx,\maxy) -- cycle; 
   \draw[draw=none,pattern=north west lines, pattern color=lightgray]  (\minx, {3 - (\minx-0)*16})  -- (0, 3) -- (\maxx,{3 -(\maxx-0)*8}) -- (\maxx,\maxy) -- (\minx,\maxy) -- cycle; 
   \draw[draw=none,fill=lightgray,opacity=0.5] ({5/6 - (\maxy - (-13/3))/16} ,\maxy) -- (3/4,-3) -- (\maxx,{3 -(\maxx-0)*8}) -- (\maxx,\maxy) -- cycle; 
   
   \draw ({5/6 - (\maxy - (-13/3))/16} ,\maxy) -- (5/6, -13/3) node[sloped, above, pos=0.4] {$\mathcal{V}^{a\mathrm{u}}_1$} -- (\maxx,{-13/3 -(\maxx-5/6)*10});
    \draw (\minx, {3 - (\minx-0)*16})  -- (0, 3) -- (\maxx,{3 -(\maxx-0)*8}) node[sloped, above,pos=0.2] {$\mathcal{Y}^{a\mathrm{u}}_1$};
   \draw[thick] ({5/6 - (\maxy - (-13/3))/16} ,\maxy) -- (3/4,-3) -- (\maxx,{3 -(\maxx-0)*8});
    \draw (\maxx,\maxy) node[below left] {$\mathcal{V}^{a\mathrm{u}}_1\cap\mathcal{Y}^{a\mathrm{u}}_1$};
  
     \foreach \y/\yname in {{-13/3}/-\tfrac{13}{3},9,-3,3}
        \draw (0,\y) node[left] {\footnotesize$\yname$} -- ++(1ex,0);    
 
      \foreach \x/\xname in {{3/4}/\tfrac{3}{4},{5/6}/\tfrac{5}{6},{3/8}/\tfrac{3}{8}}
        \draw (\x,0) node[below] {\footnotesize$\xname$} -- ++(0,1ex);         

      \foreach \x/\xname in {{9/16}/\tfrac{9}{16}}
        \draw (\x,0) -- ++(0,1ex) node[above] {\footnotesize$\xname$};
        
    \draw (\minx,0) -- (\maxx,0) node[above left] {$x^1$};
    \draw (0,\maxy) -- (0,\miny) node[above right] {$x^2$};
 \end{tikzpicture}
 &
   \footnotesize\begin{tikzpicture}[x=0.5/(\maxx-(\minx))*\figurewidth,y=0.5/(\maxy-(\miny))*\figureheight]
   \draw[draw=none,pattern=north east lines, pattern color=lightgray]  ({5/6 - (\maxy - (-13/3))/16} ,\maxy) -- (3/4,-3) -- (\maxx,{3 -(\maxx-0)*8}) -- (\maxx,\maxy) -- cycle; 
   \draw[draw=none,pattern=north west lines, pattern color=lightgray]  (\minx, {4 - (\minx-0)*16})  -- (0, 4) -- (\maxx,{4 -(\maxx-0)*8}) -- (\maxx,\maxy) -- (\minx,\maxy) -- cycle; 
   \draw[draw=none,fill=lightgray,opacity=0.5] (\minx, {4 - (\minx-0)*16})  -- (0, 4) -- (5/8,-1) -- (3/4,-3) -- (\maxx,{3 -(\maxx-0)*8}) -- (\maxx,\maxy) -- (\minx,\maxy) -- cycle; 
   
   \draw ({5/6 - (\maxy - (-13/3))/16} ,\maxy) -- (3/4,-3)  node[sloped, above, pos=0.4] {$\mathcal{V}^{a\mathrm{u}}_1\cap\mathcal{Y}^{a\mathrm{u}}_1$} -- (\maxx,{3 -(\maxx-0)*8});
   \draw (\minx, {4 - (\minx-0)*16})  -- (0, 4) -- (\maxx,{4 -(\maxx-0)*8}) node[sloped, above,pos=0.2] {$\mathcal{X}^{a\mathrm{u}}_1$};
   \draw[thick] (\minx, {4 - (\minx-0)*16})  -- (0, 4) -- (5/8,-1) -- (3/4,-3) -- (\maxx,{3 -(\maxx-0)*8});
   \draw (\maxx,\maxy) node[below left] {$\mathcal{Z}^{a\mathrm{u}}_1=\left[\mathcal{V}^{a\mathrm{u}}_1\cap\mathcal{Y}^{a\mathrm{u}}_1\right]\cup\mathcal{X}^{a\mathrm{u}}_1$};
  
   \draw (0,4) node {$\bullet$};
    \draw (0,4) node[above right] {$\hat{y}_1$};
 
     \foreach \y/\yname in {9,-3,4,-1}
        \draw (0,\y) node[left] {\footnotesize$\yname$} -- ++(1ex,0);    
 
      \foreach \x/\xname in {{3/4}/\tfrac{3}{4},{1/2}/\tfrac{1}{2}}
        \draw (\x,0) node[below] {\footnotesize$\xname$} -- ++(0,1ex);         

              \foreach \x/\xname in {{9/16}/\tfrac{9}{16},{5/8}/\tfrac{5}{8}}
        \draw (\x,0) -- ++(0,1ex) node[above] {\footnotesize$\xname$};

    \draw (\minx,0) -- (\maxx,0) node[above left] {$x^1$};
    \draw (0,\maxy) -- (0,\miny) node[above right] {$x^2$};
 \end{tikzpicture}
   \end{tabular}
    \caption{$\mathcal{W}^{a\mathrm{u}}_1$, $\mathcal{V}^{a\mathrm{u}}_1$, $\mathcal{V}^{a\mathrm{u}}_1\cap\mathcal{Y}^{a\mathrm{u}}_1$, $\mathcal{Z}^{a\mathrm{u}}_1$ and $\hat{y}_1$, Example \ref{ex:toy}}
    \label{fig:ex:toy:Zau1}
   \end{figure}
 
 Finally, following the same steps for time $t=0$ results in the sets of portfolios $\mathcal{W}^a_0$, $\mathcal{V}^a_0$ and $\mathcal{Z}^a_0$ depicted in Figure~\ref{fig:ex:toy:Z0}. The ask prices of the game option in the two currencies are the intersections of the lower boundary of $\mathcal{Z}^a_0$ with the axes, and so can be read directly from the final graph in Figure \ref{fig:ex:toy:Z0}, namely
 \begin{align*}
  \pi^a_1(Y,X,X') &= \tfrac{2}{5}, & \pi^a_2(Y,X,X') &= 4.
 \end{align*}

 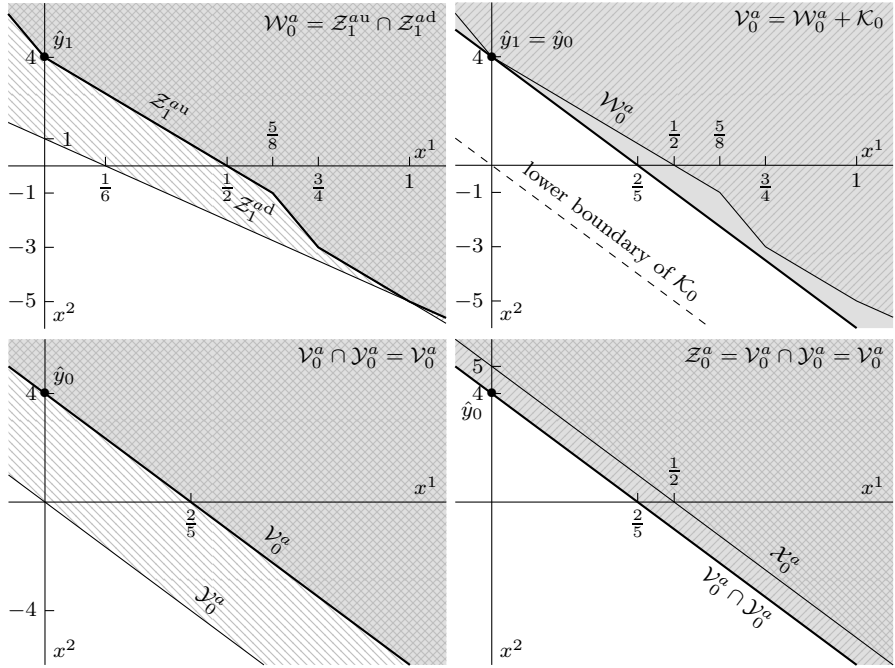
\begin{figure}
  \newcommand{\minx}{-0.1}
  \newcommand{\maxx}{1.1}
  \newcommand{\miny}{-6}
  \newcommand{\maxy}{6}
 \centering
  \begin{tabular}{@{}r@{}r@{}}
    \footnotesize\begin{tikzpicture}[x=0.5/(\maxx-(\minx))*\figurewidth,y=0.5/(\maxy-(\miny))*\figureheight]
   \draw[draw=none,pattern=north east lines, pattern color=lightgray]  (\minx, {4 - (\minx-0)*16})  -- (0, 4) -- (5/8,-1) -- (3/4,-3) -- (\maxx,{3 -(\maxx-0)*8}) -- (\maxx,\maxy) -- (\minx,\maxy) -- cycle; 
   \draw[draw=none,pattern=north west lines, pattern color=lightgray]  (\minx, {1 - (\minx-0)*6})  -- (0, 1) -- (\maxx,{1 -(\maxx-0)*6}) -- (\maxx,\maxy) -- (\minx,\maxy) -- cycle; 
   \draw[draw=none,fill=lightgray,opacity=0.5] (\minx, {4 - (\minx-0)*16})  -- (0, 4) -- (5/8,-1) -- (3/4,-3) -- (1,-5) -- (\maxx,{1 -(\maxx-0)*6}) -- (\maxx,\maxy) -- (\minx,\maxy) -- cycle; 
   
   \draw (\minx, {4 - (\minx-0)*16})  -- (0, 4) -- (5/8,-1)  node[sloped, above,midway] {$\mathcal{Z}^{a\mathrm{u}}_1$}-- (3/4,-3) -- (\maxx,{3 -(\maxx-0)*8});
   \draw (\minx, {1 - (\minx-0)*6})  -- (0, 1) -- (\maxx,{1 -(\maxx-0)*6}) node[sloped, above,midway] {$\mathcal{Z}^{a\mathrm{d}}_1$};
   \draw[thick] (\minx, {4 - (\minx-0)*16})  -- (0, 4) -- (5/8,-1) -- (3/4,-3) -- (1,-5) -- (\maxx,{1 -(\maxx-0)*6});
   \draw (\maxx,\maxy) node[below left] {$\mathcal{W}^a_0=\mathcal{Z}^{a\mathrm{u}}_1\cap\mathcal{Z}^{a\mathrm{d}}_1$};
  
   \draw (0,4) node {$\bullet$};
   \draw (0,4) node[above right] {$\hat{y}_1$};
    
    \foreach \y/\yname in {-3,4,-1,-5}
        \draw (0,\y) node[left] {\footnotesize$\yname$} -- ++(1ex,0);    
 
    \foreach \y/\yname in {1}
        \draw (0,\y) -- ++(1ex,0) node[right] {\footnotesize$\yname$};    
 
      \foreach \x/\xname in {{3/4}/\tfrac{3}{4},{1/2}/\tfrac{1}{2},1,{1/6}/\tfrac{1}{6}}
        \draw (\x,0) node[below] {\footnotesize$\xname$} -- ++(0,1ex);         

      \foreach \x/\xname in {{5/8}/\tfrac{5}{8}}
        \draw (\x,0) -- ++(0,1ex) node[above] {\footnotesize$\xname$};         

    \draw (\minx,0) -- (\maxx,0) node[above left] {$x^1$};
    \draw (0,\maxy) -- (0,\miny) node[above right] {$x^2$};
 \end{tikzpicture}  
 &
     \footnotesize\begin{tikzpicture}[x=0.5/(\maxx-(\minx))*\figurewidth,y=0.5/(\maxy-(\miny))*\figureheight]
   \draw[draw=none,pattern=north east lines, pattern color=lightgray]  (\minx, {4 - (\minx-0)*16})  -- (0, 4) -- (5/8,-1) -- (3/4,-3) -- (1,-5) -- (\maxx,{1 -(\maxx-0)*6}) -- (\maxx,{3 -(\maxx-0)*8}) -- (\maxx,\maxy) -- (\minx,\maxy) -- cycle; 
   \draw[draw=none,fill=lightgray,opacity=0.5] (\minx, {4 - (\minx-0)*10})  -- (0, 4) -- ({0-(\miny-4)/10)},\miny) -- (\maxx,\miny) -- (\maxx,\maxy) -- (\minx,\maxy) -- cycle; 
   
   \draw[thick] (\minx, {4 - (\minx-0)*10})  -- (0, 4) -- ({0-(\miny-4)/10)},\miny);
   \draw[dashed] (\minx, {0 - (\minx-0)*10})  -- (0, 0) -- ({0-(\miny-0)/10)},\miny) node[sloped, above, midway] {lower boundary of $\mathcal{K}_0$};

   \draw (\minx, {4 - (\minx-0)*16})  -- (0, 4) -- (5/8,-1)  node[sloped, above,midway] {$\mathcal{W}^a_0$}-- (3/4,-3) -- (1,-5) -- (\maxx,{1 -(\maxx-0)*6});
   \draw (\maxx,\maxy) node[below left] {$\mathcal{V}^a_0=\mathcal{W}^a_0+\mathcal{K}_0$};
  
  \draw (0,4) node {$\bullet$};
  \draw (0,4) node[above right] {$\hat{y}_1=\hat{y}_0$};

  \foreach \y/\yname in {-3,4,-1,-5}
        \draw (0,\y) node[left] {\footnotesize$\yname$} -- ++(1ex,0);    
 
      \foreach \x/\xname in {{3/4}/\tfrac{3}{4},1,{2/5}/\tfrac{2}{5}}
        \draw (\x,0) node[below] {\footnotesize$\xname$} -- ++(0,1ex);         

       \foreach \x/\xname in {{1/2}/\tfrac{1}{2},{5/8}/\tfrac{5}{8}}
        \draw (\x,0) -- ++(0,1ex) node[above] {\footnotesize$\xname$};         

   \draw (\minx,0) -- (\maxx,0) node[above left] {$x^1$};
    \draw (0,\maxy) -- (0,\miny) node[above right] {$x^2$};
 \end{tikzpicture} 
 \\
 \footnotesize\begin{tikzpicture}[x=0.5/(\maxx-(\minx))*\figurewidth,y=0.5/(\maxy-(\miny))*\figureheight]
   \draw[draw=none,pattern=north east lines, pattern color=lightgray]  (\minx, {4 - (\minx-0)*10})  -- (0, 4) -- ({0-(\miny-4)/10)},\miny) -- (\maxx,\miny) -- (\maxx,\maxy) -- (\minx,\maxy) -- cycle; 
   \draw[draw=none,pattern=north west lines, pattern color=lightgray]  (\minx, {0 - (\minx-0)*10})  -- (0, 0) -- ({0-(\miny-0)/10)},\miny) -- (\maxx,\miny) -- (\maxx,\maxy) -- (\minx,\maxy) -- cycle; 
   \draw[draw=none,fill=lightgray,opacity=0.5]  (\minx, {4 - (\minx-0)*10})  -- (0, 4) -- ({0-(\miny-4)/10)},\miny) -- (\maxx,\miny) -- (\maxx,\maxy) -- (\minx,\maxy) -- cycle; 
  
   \draw (\minx, {0 - (\minx-0)*10})  -- (0, 0) -- ({0-(\miny-0)/10)},\miny) node[sloped, above,pos=0.7] {$\mathcal{Y}^a_0$};
   \draw[thick] (\minx, {4 - (\minx-0)*10})  -- (0, 4) -- ({0-(\miny-4)/10)},\miny) node[sloped, above,pos=0.6] {$\mathcal{V}^a_0$};
   \draw (\maxx,\maxy) node[below left] {$\mathcal{V}^a_0\cap\mathcal{Y}^a_0=\mathcal{V}^a_0$};
 
   \draw (0,4) node {$\bullet$};
   \draw (0,4) node[above right] {$\hat{y}_0$};

     \foreach \y/\yname in {4,-4}
        \draw (0,\y) node[left] {\footnotesize$\yname$} -- ++(1ex,0);    
 
      \foreach \x/\xname in {{2/5}/\tfrac{2}{5}}
        \draw (\x,0) node[below] {\footnotesize$\xname$} -- ++(0,1ex);         

    \draw (\minx,0) -- (\maxx,0) node[above left] {$x^1$};
    \draw (0,\maxy) -- (0,\miny) node[above right] {$x^2$};
 \end{tikzpicture} 
 &
   \footnotesize\begin{tikzpicture}[x=0.5/(\maxx-(\minx))*\figurewidth,y=0.5/(\maxy-(\miny))*\figureheight]
   \draw[draw=none,pattern=north east lines, pattern color=lightgray]  (\minx, {4 - (\minx-0)*10})  -- (0, 4) -- ({0-(\miny-4)/10)},\miny) -- (\maxx,\miny) -- (\maxx,\maxy) -- (\minx,\maxy) -- cycle; 
   \draw[draw=none,pattern=north west lines, pattern color=lightgray]  (\minx, {5 - (\minx-0)*10})  -- (0, 5) -- (\maxx,{5 -(\maxx-0)*10}) -- (\maxx,\miny) -- (\maxx,\maxy) -- (\minx,\maxy) -- cycle; 
   \draw[draw=none,fill=lightgray,opacity=0.5] (\minx, {4 - (\minx-0)*10})  -- (0, 4) -- ({0-(\miny-4)/10)},\miny) -- (\maxx,\miny) -- (\maxx,\maxy) -- (\minx,\maxy) -- cycle; 
   
   \draw[thick] (\minx, {4 - (\minx-0)*10})  -- (0, 4) -- ({0-(\miny-4)/10)},\miny) node[sloped, below,pos=0.7] {$\mathcal{V}^a_0\cap\mathcal{Y}^a_0$};
   \draw (\minx, {5 - (\minx-0)*10})  -- (0, 5) -- (\maxx,{5 -(\maxx-0)*10}) node[sloped, above,pos=0.7] {$\mathcal{X}^a_0$};
   \draw (\maxx,\maxy) node[below left] {$\mathcal{Z}^a_0=\mathcal{V}^a_0\cap\mathcal{Y}^a_0=\mathcal{V}^a_0$};
  
   \draw (0,4) node {$\bullet$};
   \draw (0,4) node[below left] {$\hat{y}_0$};

     \foreach \y/\yname in {4,5}
        \draw (0,\y) node[left] {\footnotesize$\yname$} -- ++(1ex,0);    
 
      \foreach \x/\xname in {{2/5}/\tfrac{2}{5}}
        \draw (\x,0) node[below] {\footnotesize$\xname$} -- ++(0,1ex);         

       \foreach \x/\xname in {{1/2}/\tfrac{1}{2}}
        \draw (\x,0) -- ++(0,1ex) node[above] {\footnotesize$\xname$};         

     \draw (\minx,0) -- (\maxx,0) node[above left] {$x^1$};
    \draw (0,\maxy) -- (0,\miny) node[above right] {$x^2$};
 \end{tikzpicture}

 \end{tabular}
    \caption{$\mathcal{W}^a_0$, $\mathcal{V}^a_0$, $\mathcal{V}^a_0\cap\mathcal{Y}^a_0$, $\mathcal{Z}^a_0$ and $\hat{y}_0$, Example \ref{ex:toy}}
    \label{fig:ex:toy:Z0}
 \end{figure}
 
 Consider now the buyer's case. Observe in the context of Remark~\ref{rem:negative} that $(-X,-Y,-X')$ satisfies the conditions of Definition~\ref{def:game option}, even though some of the payoff portfolios are not solvent, for example, $-X_0 = (0,-5)\notin\mathcal{K}_0$.
 
 Construction \ref{const:buyer} and similar calculations as in the seller's case yield
  \begin{align*}
  \mathcal{Z}^{b\mathrm{uu}}_2 &= \left\{\twovector{x^1}{x^2}\in\mathbb{R}^2:16x^1 + x^2 \ge -9\right\},\\
  \mathcal{Z}^{b\mathrm{ud}}_2 = \mathcal{Z}^{b\mathrm{du}}_2 &= \left\{\twovector{x^1}{x^2}\in\mathbb{R}^2:10x^1 + x^2 \ge -4\right\},\\
  \mathcal{Z}^{b\mathrm{dd}}_2 &= \left\{\twovector{x^1}{x^2}\in\mathbb{R}^2:4x^1 + x^2 \ge 0\right\},\\
  \mathcal{Z}^{b\mathrm{u}}_1 &=  \left\{\twovector{x^1}{x^2}\in\mathbb{R}^2:16x^1 + x^2 \ge -4,8x^1 + x^2 \ge -4\right\},\\
  \mathcal{Z}^{b\mathrm{d}}_1 &= \left\{\twovector{x^1}{x^2}\in\mathbb{R}^2:6x^1 + x^2 \ge -1\right\},\\
  \mathcal{Z}^b_0 &= \left\{\twovector{x^1}{x^2}\in\mathbb{R}^2:6x^1 + x^2 \ge -\tfrac{11}{5}\right\}
 \end{align*}
 together with the bid prices
 \begin{align*}
  \pi^b_1(Y,X,X') &= \tfrac{11}{50}, & \pi^b_2(Y,X,X') &= \tfrac{11}{5}.
 \end{align*}
 
  Let us use Construction \ref{constr:buyer:hedging} to find an optimal hedging strategy $(\check{\tau},\check{y})$ for the buyer in the scenario $\mathrm{uu}$ from the initial value $\check{y}_0:=\twovector{0}{-\tfrac{11}{5}}$. Clearly $\check{y}_0\notin\mathcal{X}^b_0=\mathcal{K}_0$ and so $\check{\tau}>0$. As can be seen from Figure \ref{ex:toy},
 \[
  \mathcal{W}^{b}_0\cap\left[\check{y}_0-\mathcal{K}_0\right] = \left\{\twovector{-\tfrac{3}{10}}{\tfrac{4}{5}}\right\}
 \]
 and so $\check{y}_1=\twovector{-\tfrac{3}{10}}{\tfrac{4}{5}}$ by \eqref{eq:constr:buyer:hedging}. Observing that
 \[
  \check{y}_1\notin\mathcal{X}^{b\mathrm{u}}_1=\left\{\twovector{x^1}{x^2}\in\mathbb{R}^2:16x^1 + x^2 \ge -3,8x^1 + x^2 \ge -3\right\},
 \]
 in other words that $\check{\tau}(\mathrm{uu})>1$, we proceed to select $\check{y}^\mathrm{u}_2$. Figure \ref{ex:toy} also shows that
 \[\mathcal{W}^{b\mathrm{u}}_1\cap\left[\check{y}_1-\mathcal{K}^\mathrm{u}_1\right] = \conv\left\{\check{y}_1,\twovector{-\tfrac{37}{40}}{\tfrac{29}{5}}, \twovector{-\tfrac{5}{6}}{\tfrac{13}{3}}, \twovector{0}{-4}\right\},\]
 so there is some freedom in the selection of $\check{y}^\mathrm{u}_2$. In fact
 \[
  \mathcal{W}^{b\mathrm{u}}_1\cap\left[\check{y}_1-\mathcal{K}^\mathrm{u}_1\right] \subseteq \mathcal{X}^{b\mathrm{uu}}_2=\left\{\twovector{x^1}{x^2}\in\mathbb{R}^2:16x^1 + x^2 \ge -9\right\},
 \]
 which means that every choice of $\check{y}^\mathrm{u}_2\in\mathcal{W}^{b\mathrm{u}}_1\cap\left[\check{y}_1-\mathcal{K}^\mathrm{u}_1\right]$ would enable the buyer to exercise the option at time $2$ at the node $\mathrm{uu}$ and remain in (or return to) a solvent position. It is also clear that $\check{\tau}(\mathrm{uu})=2$.
 
 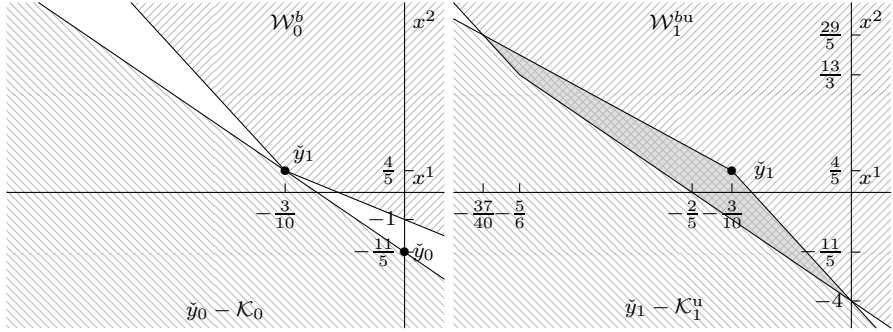
\begin{figure}
   \newcommand{\minx}{-1}
  \newcommand{\maxx}{0.1}
  \newcommand{\miny}{-5}
  \newcommand{\maxy}{7}
  \centering
  \begin{tabular}{@{}r@{}r@{}}
   \footnotesize\begin{tikzpicture}[x=0.5/(\maxx-(\minx))*\figurewidth,y=0.5/(\maxy-(\miny))*\figureheight]
   \draw[draw=none,pattern=north east lines, pattern color=lightgray]  ({-3/10-(\maxy-4/5)/16},\maxy) -- (-3/10, 4/5) -- (\maxx,{4/5 -(\maxx-(-3/10))*6}) -- (\maxx,\maxy) -- cycle; 
   \draw[draw=none,pattern=north west lines, pattern color=lightgray]  ({0-(\maxy-(-11/5))/10},\maxy)  -- (0, -11/5) -- (\maxx,{-11/5 -(\maxx-0)*10}) -- (\maxx,\miny) -- (\minx,\miny) -- (\minx,\maxy) -- cycle; 
   
   \draw ({-3/10-(\maxy-4/5)/16},\maxy) -- (-3/10, 4/5) -- (\maxx,{4/5 -(\maxx-(-3/10))*6});
   \draw ({0.5*(-3/10-(\maxy-4/5)/16+\maxx)},\maxy) node[below] {$\mathcal{W}^b_0$};
  
  \draw ({0-(\maxy-(-11/5))/10},\maxy)  -- (0, -11/5) -- (\maxx,{-11/5 -(\maxx-0)*10});
   \draw ({0.5*(\minx+\maxx)},\miny) node[above] {$\check{y}_0-\mathcal{K}_0$};
  
  \draw (0,-11/5) node {$\bullet$}; \draw (0,-11/5) node[right] {$\check{y}_0$};
  \draw (-3/10, 4/5) node {$\bullet$}; \draw (-3/10, 4/5) node[above right] {$\check{y}_1$};
  
     \foreach \y/\yname in {{4/5}/\tfrac{4}{5},{-11/5}/-\tfrac{11}{5},-1}
        \draw (0,\y) node[left] {\footnotesize$\yname$} -- ++(1ex,0);    
 
      \foreach \x/\xname in {-{3/10}/-\tfrac{3}{10}\phantom-}
        \draw (\x,0) node[below] {\footnotesize$\xname$} -- ++(0,1ex);         

    \draw (\minx,0) -- (\maxx,0) node[above left] {$x^1$};
    \draw (0,\maxy) node[below right] {$x^2$} -- (0,\miny);
 \end{tikzpicture}
 &
    \footnotesize\begin{tikzpicture}[x=0.5/(\maxx-(\minx))*\figurewidth,y=0.5/(\maxy-(\miny))*\figureheight]
   \draw[draw=none,pattern=north east lines, pattern color=lightgray]  (\minx, {13/3 - (\minx-(-5/6))*16})  -- (-5/6, 13/3) -- (\maxx,{13/3 -(\maxx-(-5/6))*10}) -- (\maxx,\maxy) -- (\minx,\maxy) -- cycle; 
   \draw[draw=none,pattern=north west lines, pattern color=lightgray]  (\minx, {4/5 - (\minx-(-3/10))*8})  -- (-3/10, 4/5) -- ({-3/10-(\miny-4/5)/16},\miny) -- (\minx,\miny) -- cycle; 
   \draw[draw=none,fill=lightgray,opacity=0.5] (-37/40,29/5) -- (-5/6, 13/3)  -- (0, -4) -- (-3/10, 4/5) -- cycle; 
   
   \draw (\minx, {13/3 - (\minx-(-5/6))*16}) -- (-5/6, 13/3) -- (\maxx,{13/3 -(\maxx-(-5/6))*10});
   \draw ({0.5*(\minx+\maxx)},\maxy) node[below] {$\mathcal{W}^{b\mathrm{u}}_1$};
   \draw (\minx, {4/5 - (\minx-(-3/10))*8})  -- (-3/10, 4/5) -- ({-3/10-(\miny-4/5)/16},\miny);
   \draw ({0.5*(\minx+(-3/10-(\miny-4/5)/16}, \miny) node[above] {$\check{y}_1-\mathcal{K}^\mathrm{u}_1$};
  
  \draw (-3/10, 4/5) node[label=right:{$\check{y}_1$}] {$\bullet$};
  
     \foreach \y/\yname in {{4/5}/\tfrac{4}{5},{29/5}/\tfrac{29}{5},{13/3}/\tfrac{13}{3},{-11/5}/-\tfrac{11}{5},-4}
        \draw (0,\y) node[left] {\footnotesize$\yname$} -- ++(1ex,0);    
 
      \foreach \x/\xname in {-{3/10}/-\tfrac{3}{10}\phantom-,-{5/6}/-\tfrac{5}{6}\phantom-,-{2/5}/-\tfrac{2}{5}\phantom-,-{37/40}/-\tfrac{37}{40}\phantom-}
        \draw (\x,0) node[below] {\footnotesize$\xname$} -- ++(0,1ex);         

    \draw (\minx,0) -- (\maxx,0) node[above left] {$x^1$};
    \draw (0,\maxy) node[below right] {$x^2$} -- (0,\miny);
 \end{tikzpicture}
  \end{tabular}
  \caption{$\check{y}_0$ and $\check{y}_1$, Example \ref{ex:toy}}
 \end{figure}

 Returning to the seller, it is straightforward to use Construction \ref{constr:writer:hedging} to create an optimal hedging strategy $(\hat{\sigma},\hat{y})$ for the seller from the initial endowment $\hat{y}_0:=\twovector{0}{4}$. From Figure \ref{fig:ex:toy:Z0} it is clear that $\hat{y}_0\notin\mathcal{X}^a_0$ (and so $\hat{\sigma}>0$), and also that
 \[
  \mathcal{W}^a_0\cap\left[\hat{y}_0-\mathcal{K}_0\right] = \left\{\twovector{0}{4}\right\} = \left\{\hat{y}_0\right\},
 \]
 and so by \eqref{eq:constr:seller:hedging} we must choose $\hat{y}_1:=\hat{y}_0$. Figure \ref{fig:ex:toy:Zau1} and \eqref{eq:ex:Zad1} shows that $\hat{y}_1\notin\mathcal{X}^a_1$ and so $\hat{\sigma}=1$. This means that if the seller follows $(\hat{\sigma},\hat{y})$ and the buyer $(\check{\tau},\check{y})$, then in scenario $\mathrm{uu}$ the option will be cancelled (but not exercised) at the node~$\mathrm{u}$, at which time the seller delivers $\twovector{0}{4}$ to the buyer.
\end{example}

\begin{example}\label{ex:counter}
Figure~\ref{fig:one-step-model} presents a game option $(Y,X,Y)$ in a binary single-step two-currency model. The implied assumption $X'=Y$ is conventional in the game options literature (cf.~Remark \ref{rem:convention}). Note also that the option satisfies Definition~\ref{def:game option} despite the fact that neither $Y_0$ nor $X_0$ are solvent (cf.~Remark~\ref{rem:negative}). 
 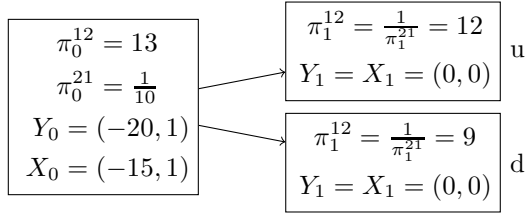
\begin{figure}
 \begin{center}
 \begin{tikzpicture}[model,max label lines = 2.5]
  \node (root) {$\begin{gathered}\begin{aligned}\pi^{12}_0 &= 13\\\pi^{21}_0 &= \tfrac{1}{10}\end{aligned}\\\begin{aligned}Y_0&=\twovector{-20}{1} \\ X_0&=\twovector{-15}{1}\end{aligned}\end{gathered}$}
    child { node[label=right:u] (u) {$\begin{gathered}\pi^{12}_1 =\tfrac{1}{\pi^{21}_1} = 12\\
	Y_1=X_1=\twovector{0}{0}\end{gathered}$}
	}
    child { node[label=right:d] (d) {$\begin{gathered}\pi^{12}_1 =\tfrac{1}{\pi^{21}_1} = 9\\
	Y_1=X_1=\twovector{0}{0}\end{gathered}$}
	};
\end{tikzpicture}
\end{center}
  \caption{Game option in binary single-step two-currency model, Example \ref{ex:counter}}
  \label{fig:one-step-model}
 \end{figure}
 
 It would be tempting for the seller to cancel the option at time $0$, because delivering $X_0$ to the buyer at time $0$ is effectively the same as receiving the portfolio $\twovector{15}{-1}$, which the seller could immediately convert into
 \[
  15 - \pi^{12}_0 = 2 > 0
 \]
 units of currency $1$. If the buyer exercises at time $0$ then the seller is in an even better position, because $-Y_0$ could similarly be converted into $7$ units of currency $1$. Applying Construction \ref{constr:writer} to this option gives the set $\mathcal{Z}^a_0$ presented in Figure \ref{fig:ex:counter:Z0}. Clearly
 \[
  \pi^a_1(Y,X,Y)=-2,
 \]
 which means that cancellation at time $0$ is indeed optimal for the seller. Theorem~3.2(i) of \citet{Kifer2013a} constructs a function $z_0$, also shown in Figure~\ref{fig:ex:counter:Z0}, and gives the ask price as
 \[
  z_0(0) = -2 = \pi^a_1(Y,X,Y).
 \]
 Note that $\mathcal{Z}^a_0$ is the epigraph of $z_0$, reflected around the line $x^1=x^2$; in other words, $\mathcal{Z}^a_0$ would have been the epigraph for $z_0$ had the currencies had been ordered differently.
 \begin{figure}
  \newcommand{\minx}{-20}
  \newcommand{\maxx}{90}
  \newcommand{\miny}{-8}
  \newcommand{\maxy}{2}
  \centering
\begin{tabular}{@{}r@{}r@{}}
   \footnotesize\begin{tikzpicture}[x=0.5/(\maxx-(\minx))*\figurewidth,y=0.5/(\maxy-(\miny))*\figurewidth]
   \draw[draw=none,fill=lightgray,opacity=0.5] (\maxx, {-7 - (\maxx-84)/13}) -- (84,-7) -- (25,-27/13) -- (-15,1) -- (\minx, {1 - (\minx+15)/10}) -- (\minx,\maxy) -- (\maxx,\maxy) -- (\maxx,\miny) -- cycle; 
   
   \draw[thick]  (\maxx, {-7 - (\maxx-84)/13}) -- (84,-7) -- (25,-27/13) -- (-15,1) -- (\minx, {1 - (\minx+15)/10});

   \draw (\maxx,\maxy) node[below left] {$\mathcal{Z}^a_0$};
 
      \foreach \x/\xname in {-15/-15\phantom-,-2/-2\phantom-,25,84}
        \draw (\x,0) node[below] {\footnotesize$\xname$} -- ++(0,1ex);         

     \foreach \y/\yname in {-7/-7,{-27/13}/-\frac{27}{13},1}
        \draw (0,\y) node[left] {\footnotesize$\yname$} -- ++(1ex,0);    
 
     \draw (\minx,0) -- (\maxx,0) node[above left] {$x^1$};
    \draw (0,\maxy) -- (0,\miny) node[above right] {$x^2$};
 \end{tikzpicture} &
   \renewcommand{\minx}{-8}
  \renewcommand{\maxx}{2}
  \renewcommand{\miny}{-20}
  \renewcommand{\maxy}{90}
   \footnotesize\begin{tikzpicture}[x=0.5/(\maxx-(\minx))*\figurewidth,y=0.5/(\maxy-(\miny))*\figurewidth]

   \draw[thick] ({-7 - (\maxy-84)/13},\maxy) -- (-7,84) node[right,midway] {$z_0$} -- (-27/13,25) -- (1,-15) -- ({1 - (\miny+15)/10},\miny);

      \foreach \x/\xname in {-15,-2,25,84}
        \draw (0,\x) node[left] {\footnotesize$\xname$} -- ++(1ex,0);    

     \foreach \y/\yname in {-7/-7\phantom-,{-27/13}/-\frac{27}{13}\phantom-,1}
        \draw (\y,0) node[below] {\footnotesize$\yname$} -- ++(0,1ex);         
 
     \draw (\minx,0) -- (\maxx,0) node[above left] {$x^2$};
    \draw (0,\maxy) -- (0,\miny) node[above right] {$x^1$};
 \end{tikzpicture}
 \end{tabular}
    \caption{$\mathcal{Z}^a_0$ and $z_0$, Example \ref{ex:counter}}
    \label{fig:ex:counter:Z0}
 \end{figure}
 
 We will compare the dual representations for $\pi^a_1(Y,X,Y)$ in Theorem~\ref{th:ask-price} above and Theorem~3.1 of \citet{Kifer2013a}. The model has only two stopping times, $0$ and $1$, and so this can be done easily and directly. Observe also that
 \begin{align*}
  \mathcal{K}^\ast_0 &= \cone \left\{\twovector{1}{10},\twovector{1}{13}\right\}, \\
  \mathcal{K}^{\ast\mathrm{u}}_1 &= \cone \left\{\twovector{1}{12}\right\} \subset\mathcal{K}^\ast_0, \\
  \mathcal{K}^{\ast\mathrm{d}}_1  &= \cone \left\{\twovector{1}{9}\right\}.
 \end{align*}
 This means (cf.~Definition~\ref{def:approx-martingale}) that the property $(\mathbb{P},S)=\left(\mathbb{P},\twovector{S^1}{S^2}\right)\in\bar{\mathcal{P}}^1(\chi)$ is equivalent to
 \begin{align*}
 \mathbb{P}(\mathrm{u}) &\ge \tfrac{1}{3}, & \mathbb{P}(\mathrm{d}) &= 1- \mathbb{P}(\mathrm{u})\ge 0,
 \end{align*}
 together with
 \begin{align*}
  S^1_0&=S^1_1 = 1, & 10 &\le S^2_0 \le 13, & S^2_1(\mathrm{u})&=12, & S^2_1(\mathrm{d})&=9.
 \end{align*}

 For the representation for $\pi^a_1(Y,X,Y)$ in Theorem~\ref{th:ask-price}, take first $\sigma=0$. For every $\chi\in\mathcal{X}$ and $(\mathbb{P},S)\in\bar{\mathcal{P}}^1(\chi\wedge0)$ we have
 \begin{align*}
  (Q_{0\cdot}\cdot S_{0\wedge\cdot})_\chi 
  &= \chi^\ast_1X_0\cdot S_0 + \chi_0 Y_0 \cdot S_0
  = -5\chi_0 - 15 + S_0^2,
 \end{align*}
 and therefore 
 \begin{multline*}
  \max_{\chi\in\mathcal{X}}\max_{(\mathbb{P},S)\in\bar{\mathcal{P}}^1(\chi\wedge0)}\mathbb{E}_\mathbb{P}((Q_{0\cdot}\cdot S_{0\wedge\cdot})_\chi) \\
  = \max\{-5\chi_0 - 15 + S_0^2:\chi_0\in[0,1],S_0^2\in[10,13]\}=-2.
 \end{multline*}
 Similarly, for $\sigma=1$,
 \begin{align*}
  (Q_{1\cdot}\cdot S_{1\wedge\cdot})_\chi
  &=\chi_0 Y_0 \cdot S_0 + \chi^\ast_2X_1\cdot S_1 + \chi_1 Y_1 \cdot S_1
  =\chi_0(S_0^2-20)
 \end{align*}
 and so
 \begin{multline} \label{eq:ex:counter:1}
  \max_{\chi\in\mathcal{X}}\max_{(\mathbb{P},S)\in\bar{\mathcal{P}}^1(\chi\wedge1)}\mathbb{E}_\mathbb{P}((Q_{1\cdot}\cdot S_{1\wedge\cdot})_\chi) \\
  = \max\{\chi_0(S_0^2-20):\chi_0\in[0,1],S_0^2\in[10,13]\}=0.
 \end{multline}
 Finally,
 \[
  \min_{\sigma\in\mathcal{T}}\max_{\chi\in\mathcal{X}}\max_{(\mathbb{P},S)\in\bar{\mathcal{P}}^1(\chi\wedge\sigma)}\mathbb{E}_\mathbb{P}((Q_{\sigma\cdot}\cdot S_{\sigma\wedge\cdot})_\chi) = \min\{-2,0\} = -2 = \pi^a_1(Y,X,Y)
 \]
 as expected.
 
 Now turn to the dual representation in Theorem~3.1 of \citet{Kifer2013a}. It can be written in our notation as
 \[
  V^a := \min_{\sigma\in\mathcal{T}}\max_{\chi\in\mathcal{X}}\max_{(\mathbb{P},S)\in\bar{\mathcal{P}}^1(\chi)}\mathbb{E}_\mathbb{P}((Q_{\sigma\cdot}\cdot S)_\chi),
 \]
 where
 \[
  (Q_{\sigma\cdot}\cdot S)_\chi = \sum_{t=0}^\sigma\chi_tY_t\cdot S_t + \sum_{t=\sigma+1}^T\chi_tX_\sigma\cdot S_t
\]
and $T=1$. For $\sigma=1$, the calculation \eqref{eq:ex:counter:1} gives that
\[
 \max_{\chi\in\mathcal{X}}\max_{(\mathbb{P},S)\in\bar{\mathcal{P}}^1(\chi)}\mathbb{E}_\mathbb{P}((Q_{1\cdot}\cdot S)_\chi)
 = \max_{\chi\in\mathcal{X}}\max_{(\mathbb{P},S)\in\bar{\mathcal{P}}^1(\chi\wedge1)}\mathbb{E}_\mathbb{P}((Q_{1\cdot}\cdot S_{1\wedge\cdot})_\chi) = 0.
\]
For $\sigma=0$ and every $\chi\in\mathcal{X}$ and $(\mathbb{P},S)\in\bar{\mathcal{P}}^1(\chi)$ we have
\begin{align*}
 (Q_{\sigma\cdot}\cdot S)_\chi
 &= \chi_0Y_0\cdot S_0 + \chi_1X_0\cdot S_1 = \chi_0(S^2_0-5)+(1-\chi_0)S_1^2-15,
\end{align*}
so that
\begin{align*}
 \mathbb{E}_\mathbb{P}((Q_{0\cdot}\cdot S)_\chi) = \chi_0(S^2_0-5)+(1-\chi_0)(3\mathbb{P}(\mathrm{u})+9)-15.
\end{align*}
This means that
\begin{multline*}
 \max_{\chi\in\mathcal{X}}\max_{(\mathbb{P},S)\in\bar{\mathcal{P}}^1(\chi)}\mathbb{E}_\mathbb{P}((Q_{0\cdot}\cdot S)_\chi) \\
 = \max\{\chi_0(S^2_0-5)+(1-\chi_0)(3\mathbb{P}(\mathrm{u})+9)-15\\
 :\chi_0\in[0,1],S_0^2\in[10,13],\mathbb{P}(\mathrm{u})\in[\tfrac{1}{3},1]\}  = -3.
\end{multline*}
Finally,
\[
 V^a = \min\{0,-3\}=-3\neq\pi^a_1(Y,X,Y),
\]
which demonstrates that $V^a$ is not the correct dual representation for $\pi^a_1(Y,X,Y)$.
\end{example}

\begin{example} \label{ex:three-asset}
 Consider a three-currency model with $T$ steps and time horizon $1$ based on the two-asset recombinant Korn-Muller model \citep{korn_muller2009} with Cholesky decomposition, i.e.~friction-free exchange rates in terms of currency $3$ are modelled by the process $(S_t)_{t=0}^T$, where
 \[
  S_t = \threevector{\varepsilon^1_tS^1_{t-1}}{\varepsilon^2_tS^2_{t-1}}{1} \text{ for }t=1,\ldots,T
 \]
 and $(\varepsilon_t)_{t=1}^T = \left(\twovector{\varepsilon^1_t}{\varepsilon^2_t}\right)_{t=1}^T$ is a sequence of independent identically distributed random variables taking the values
\begin{align*}
 \twovector{e^{-\frac{1}{2}\sigma_1^2\Delta-\sigma_1\sqrt\Delta}}{e^{-\frac{1}{2}\sigma_2^2\Delta-(\rho+\sqrt{1-\rho^2})\sigma_2\sqrt\Delta}},\\
  \twovector{e^{-\frac{1}{2}\sigma_1^2\Delta-\sigma_1\sqrt\Delta}}{e^{-\frac{1}{2}\sigma_2^2\Delta-(\rho-\sqrt{1-\rho^2})\sigma_2\sqrt\Delta}},\\
 \twovector{e^{-\frac{1}{2}\sigma_1^2\Delta+\sigma_1\sqrt\Delta}}{e^{-\frac{1}{2}\sigma_2^2\Delta+(\rho-\sqrt{1-\rho^2})\sigma_2\sqrt\Delta}},\\
  \twovector{e^{-\frac{1}{2}\sigma_1^2\Delta+\sigma_1\sqrt\Delta}}{e^{-\frac{1}{2}\sigma_2^2\Delta+(\rho+\sqrt{1-\rho^2})\sigma_2\sqrt\Delta}},
\end{align*}
each with positive probability, where $\Delta:=\frac{1}{T}$ is the step size. The exchange rates with transaction costs are modelled as
\[
 \pi^{ij}_t:=
 \begin{cases}
 \frac{S^j_t}{S^i_t}(1+k) &\text{if }i\neq j,\\
 1 & \text{if }i=j,
 \end{cases}
\]
for $i,j=1,\ldots,3$ and $t\le T$, where $k\in[0,1)$. We take
\begin{align*}
 \twovector{S^1_0}{S^2_0} &= \twovector{40}{50}, & \sigma_1&=0.15, & \sigma_2&=0.1, & \rho&=0.5.
\end{align*}

Consider a game put option with physical delivery on a basket containing one unit each of currencies~$1$ and $2$ and with strike $K$ in currency $3$, i.e.
\[
 Y_t=\threevector{-1}{-1}{K}\text{ for }t=0,\ldots,T.
\]
On cancellation the seller delivers the above payoff to the buyer, together with a cancellation penalty $p\ge0$ in currency $3$, so that
\[
 X_t=X'_t = \threevector{-1}{-1}{K+p}\text{ for }t=0,\ldots,T.
\]
We allow for the possibility that the seller may choose not to cancel the option, and the buyer may choose not to exercise, by adding an additional time step $T+1$ and taking
\[
 Y_{T+1} = X_{T+1} = X'_{T+1}=\threevector{0}{0}{0}.
\]

Except for the union, the operations in Constructions \ref{constr:writer} and \ref{const:buyer}, namely intersection and direct addition of a polyhedral cone, are standard geometric procedures when applied to polyhedra, and can be implemented using existing software libraries. Both these operations are union-preserving, and so the extension to unions of polyhedra is straightforward. The numerical results below were produced using Maple with the \emph{Convex} package \citep{Franz2009}. 

Table \ref{tab:ex:strikes} contains bid and ask prices in currency $3$ of the basket put with penalty $p=5$ for a range of strike prices. Both bid and ask prices increase with the strike. Note the appearance of negative bid and ask prices for out-of-the-money options (i.e.~$K<S^1_0+S^2_0=90$). The reason for this is that the seller can cancel the option at any time, and cancellation tends to be particularly attractive to the seller (and very costly for the buyer) when the option is far out of the money; for example, if $K=80$ then, from the point of view of the seller, cancelling the option at time $0$ is equivalent to receiving the basket from the buyer and paying $K+p=85$ in currency $3$, which is less than the market price $S^1_0+S^2_0=90$ (ignoring transaction costs) of the basket.
\begin{table}
 \centering
 \begin{tabular}{|r|d{6}d{6}|d{6}d{6}|}
 \hline
 & \multicolumn{2}{c|}{$k=0$} & \multicolumn{2}{c|}{$k=0.005$} \\
 $K$ & \multicolumn{1}{c}{$\pi^b_3$} & \multicolumn{1}{c|}{$\pi^a_3$} & \multicolumn{1}{c}{$\pi^b_3$} & \multicolumn{1}{c|}{$\pi^a_3$} \\
 \hline
100 & 10.043290 & 11.033942 & 9.568590 & 11.687749 \\
95 & 5.266479 & 6.817389 & 4.706543 & 7.480028 \\
90 & 0.967824 &2.774598 & 0.367975 & 3.423798 \\
85 & -2.934360 & -1.091048 & -3.587214 & -0.444708 \\
80 & -6.910514 & -5.131149 & -7.584034 & -4.614029 \\
\hline
  \end{tabular}
  \caption{Bid and ask prices of game basket put option with $N=10$, $p=5$ and different strikes $K$, Example \ref{ex:three-asset}}
  \label{tab:ex:strikes}
\end{table}

Bid and ask prices of the game basket put with strike $K=100$ and a range of penalty values are reported in Table \ref{tab:ex:penalties}, together with bid and ask prices of the American basket put with the same strike (using the constructions of \cite{Roux_Zastawniak2015}). In practical terms, game options with large penalties resemble American options (because larger penalties make it less attractive for the seller to cancel the option early), and this explains the convergence of the bid and ask prices of the game option to that of the American option as the penalty increases.
\begin{table}
 \centering
 \begin{tabular}{|r|d{6}d{6}|d{6}d{6}|}
 \hline
 & \multicolumn{2}{c|}{$k=0$} & \multicolumn{2}{c|}{$k=0.005$} \\
 $p$ & \multicolumn{1}{c}{$\pi^b_3$} & \multicolumn{1}{c|}{$\pi^a_3$} & \multicolumn{1}{c}{$\pi^b_3$} & \multicolumn{1}{c|}{$\pi^a_3$} \\
 \hline
0 & 10.000000 & 10.000000 & 9.550000 & 10.447761 \\
1 & 10.014709 & 10.278348 & 9.556726 & 10.790910 \\
2 & 10.027095 & 10.497310 & 9.562075 & 11.051351 \\
5 & 10.043290 & 11.033942 & 9.568590 & 11.687749 \\
10 & 10.050958 & 11.571315 & 9.571850 & 12.297913 \\
20 & 10.052026 & 11.796921 & 9.572414 & 12.575621 \\
\hline
American & 10.052027 & 11.812658 & 9.572414 & 12.589930 \\
\hline
  \end{tabular}
  \caption{Bid and ask prices of game basket put option with $N=10$, $K=100$ and different penalties, Example \ref{ex:three-asset}}
  \label{tab:ex:penalties}
\end{table}
\end{example}

\end{document}